\pdfoutput=1
\documentclass[aps,pra,letterpaper,10pt,twocolumn,notitlepage,superscriptaddress,floatfix]{revtex4-2}
\usepackage[T1]{fontenc}
\usepackage[utf8]{inputenc}
\usepackage{amsmath, amsthm, amsfonts, amssymb, microtype, mathtools, qcircuit, xcolor, bm, bbm, graphicx, hyperref, cleveref, tikz}
\usepackage{soul}
\hypersetup{
    colorlinks,
    linkcolor={blue},
    citecolor={blue},
    urlcolor={blue},
    pdfauthor={Robin Harper, Steven T Flammia},
    pdftitle={Learning correlated noise in a 39-qubit quantum processor}
}
\usetikzlibrary{external,matrix,automata,positioning,fit,backgrounds,shapes.geometric, arrows, shadows}
\usepackage{longtable}
\usepackage[pass]{geometry}
\usepackage{etoolbox}

\graphicspath{{figures/}}

\usepackage{booktabs}

\theoremstyle{plain}
\newtheorem{theorem}{Theorem}

\theoremstyle{definition}

\renewcommand*{\ss}[1]{{\sf{#1}}}

\newcommand{\ket}[1]{|#1\rangle}
\newcommand{\bra}[1]{\langle #1|}

\newcommand{\sccmacro}{\mathbf{\mathcal{S}}}

\newcommand{\acr}{LACE}

\begin{document}
\tikzset{ 
    table/.style={
        matrix of nodes,
        row sep=-\pgflinewidth,
        column sep=-\pgflinewidth,
        nodes={
            rectangle,
            draw=black,
            align=center
        },
        minimum height=1.3em,
        text depth=0.45ex,
        text height=1.21ex,
        nodes in empty cells,
        every even row/.style={
            nodes={fill=gray!20}
        },
        column 1/.style={
            nodes={text width=3em,font=\bfseries}
        },
        row 1/.style={
            nodes={
                fill=black,
                text=white,
                font=\bfseries
            }
        }
    }
}

\title{Learning correlated noise in a 39-qubit quantum processor}

\author{Robin Harper}
\thanks{Corresponding author: robin.harper@sydney.edu.au}
\affiliation{Centre for Engineered Quantum Systems, School of Physics, University of Sydney, Sydney, NSW 2006 Australia}

\author{Steven T. Flammia}
\affiliation{AWS Center for Quantum Computing, Pasadena, CA, USA}
\affiliation{IQIM, California Institute of Technology, Pasadena, CA, USA}

\date{\today}

\begin{abstract}
Building error-corrected quantum computers relies crucially on measuring and modeling noise on candidate devices. 
In particular, optimal error correction requires knowing the noise that occurs in the device as it executes the circuits required for error correction. 
As devices increase in size we will become more reliant on efficient models of this noise.
However, such models must still retain the information required to optimize the algorithms used for error correction.
Here we propose a method of extracting detailed information of the noise in a device running syndrome extraction circuits. 
We introduce and execute an experiment on a superconducting device using 39 of its qubits in a surface code doing repeated rounds of syndrome extraction, but omitting the mid-circuit measurement and reset. 
We show how to extract from the 20 data qubits the information needed to build noise models of various sophistication in the form of graphical models. 
These models give efficient descriptions of noise in large-scale devices and are designed to illuminate the effectiveness  of error correction against correlated noise. 
Our estimates are furthermore precise: we learn a consistent global distribution where all one- and two-qubit error rates are known to a relative error of $\pm 0.1$\%. 
By extrapolating our experimentally learned noise models towards lower error rates, we demonstrate that accurate correlated noise models are increasingly important for successfully predicting sub-threshold behavior in quantum error correction experiments.
\end{abstract}
\maketitle
\section{Introduction}

In order to fully realise the potential of quantum devices one must execute many highly accurate quantum gates~\cite{Shor1997,Cao2019,Biamonte2017,AspuruGuzik2005,Gidney2021}. 
Current multi-qubit devices have single-qubit gates with average error rates around $10^{-3}$~\cite{Grzesiak2020,Huang2019,chen2022,ballance2016}; this is orders of magnitude too high to directly execute the number of operations required for computations such as integer factoring~\cite{Gidney2021}. 
Fault-tolerant quantum error correction overcomes this by trading more physical qubits for increased logical qubit fidelity~\cite{Knill1998a,Aharonov1997,Shor1997,Kitaev1997}. 
Although there are many proposed protocols for error correction, and many recent error correction experiments
for various small codes 
\cite{ryan-anderson2021,Egan2021,Krinner2022,sundaresan2022,zhao2022}, the most widely studied proposals are based on the surface code and variants thereof~\cite{Kitaev2003,BravyiKitaev1998,fowler2012,Tuckett2018,BonillaAtaides2021}. 
Recent work has shown a decreasing logical error rate in a variant of the surface code~\cite{BonillaAtaides2021} as the distance of the code increases~\cite{google2022}.

Error correction can be dramatically improved provided one knows certain details of the noise afflicting the device~\cite{Robertson2017,BonillaAtaides2021,Nickerson2019,Tuckett2019,Chubb2021,Chubb2021a,google2022} 
and one tailors the code and decoder to this noise~\cite{Tuckett2020,BonillaAtaides2021,iolius2022,tiurev2022,Higgott2022}. 
However, most existing noise characterization techniques fall into two categories: tomographic reconstruction or strong noise-averaging (e.g.\ randomized benchmarking). 
Those based on full tomographic noise reconstruction \cite{Chuang1997,Merkel2012,Blume-Kohout2016,evan2022} struggle to scale past a few qubits. 
Techniques that strongly average the noise~\cite{Emerson2007, Knill2008, Magesan2011, Carignan-Dugas2015, Hashagen2018, Erhard2019, Helsen2019, Proctor2018, proctor21b, hines2022} attempt to gain scalability but often fail to provide enough of a nuanced picture of the noise to allow diagnostic conclusions or to suggest improvements to the device.  We review these methods in \cref{sec:background}.

Fortunately, there are techniques that reconstruct efficient descriptions of detailed Pauli noise models for large-scale systems~\cite{Flammia2019,Harper2019}. 
These methods learn the noise in the form of a \textit{graphical model}, which is a model that is flexible enough to provide faithful descriptions of correlated noise while retaining the desirable properties of being efficient to learn and allowing the tailoring of codes and decoders to relevant features of the noise. 

An example of such a graphical model is an Ising model, where a probability distribution with $2^n$ outcomes --- the equilibrium Gibbs distribution --- is specified by a Hamiltonian and a temperature. 
The individual probabilities of such a model cannot be calculated without knowing the partition function (which is not generally efficient to compute). 
However, ratios of probabilities and certain conditional probabilities can be calculated efficiently, and they often enjoy efficient sampling algorithms (using e.g.\ Metropolis sampling). 
For quantum noise that is described by such a graphical model, these features are sufficient to enable detailed modeling and prediction, as well as optimizations such as the creation of minimum weight perfect matching decoders or tensor network decoders that are tailored to the underlying noise correlations. For instance, ref.~\cite{Robeva2018} explores the duality between such graphical models and tensor networks. Accordingly, such graphical models can directly instantiate tensor networks, allowing utilisation of the contraction methods currently proposed as a means of decoding syndrome measurements \cite{Chubb2021,Chubb2021a,huang2021efficient,Farrelly2022}.

Here we experimentally demonstrate learning a comprehensive description of the Pauli noise in a 39 qubit superconducting device running the circuits required to implement surface code quantum error correction, but without mid-circuit ancilla measurement and reset. 
The noise is learned in the form of a graphical model, meaning this description is efficient and contains a globally consistent description of the errors in the device, including larger scale correlations. 

Our reconstructed noise estimates are highly accurate. 
By way of example, a bootstrap analysis (at the $2\sigma$ level) shows a maximum \textit{relative} error of $\pm 0.1\%$ on both single-qubit error and two-qubit error rates. 
We fit several graphical models having increasing expressive power to these data to see how correlations affect quantum error correction. 
We measure the logical failure rate assuming a code capacity model (i.e., no measurement and reset errors) and utilising an approximate maximum likelihood decoder (a tensor network decoder~\cite{Bravyi2014} with bond dimension 8). 
Even in the absence of measurement and reset error, the circuit noise translated to a logical error rate of at least $0.166\pm0.004$, which is worse than the observed average single-qubit error rate of $0.136\pm 0.001$. 

Our reconstructed noise models allow us to address previously difficult-to-explore questions of interest for error correction in this system.
For example, given the observed correlations, how much lower should the single-qubit error rates be before the logical failure rate is below the single-qubit bare error rate? 
To address this, we parameterize our observed noise as being generated by a continuous time evolution evolving for a finite time $T$; this allows us to retain the observed correlation structure in a way that extrapolates smoothly to zero noise as $T\to 0$. 

Perhaps surprisingly, the simplest models (which ignore correlated errors) through to the most complex models all gave approximately the same pseudo-threshold (physical error rate = logical error rate $\approx 0.1$). 
However, simpler models underestimate the logical error rate compared with the global distribution, and give widely diverging predictions of the logical failure rate as the physical error rate extrapolates towards zero. 
For instance, with an average physical error rate of $0.031$, the simpler models predicted logical error rates of $0.006\pm0.001$, whereas the logical error rate from the extrapolated global probability distribution, was over twice as much ($0.0121\pm0.002$).
In this regime, models that capture correlated errors, such as an Ising model, gave logical error rates commensurate with or higher than the global distribution $0.014-0.018$. 

Our conclusion is that models capable of capturing correlated errors are indispensable in accurately estimating the expected logical failure rate of error correction protocols executed on this device. 
Our analysis shows that, once physical error rates are low enough such that error correction might be possible, models that fail to take into account correlated errors (such as those caused by crosstalk) can potentially underestimate expected logical failure rates by a significant fraction. 
However, models that can capture such crosstalk provide more reliable estimates of the performance of the error correcting circuitry as well as providing the more nuanced information required to write custom decoders.

\begin{figure}[ht!]
    \includegraphics[width=\columnwidth]{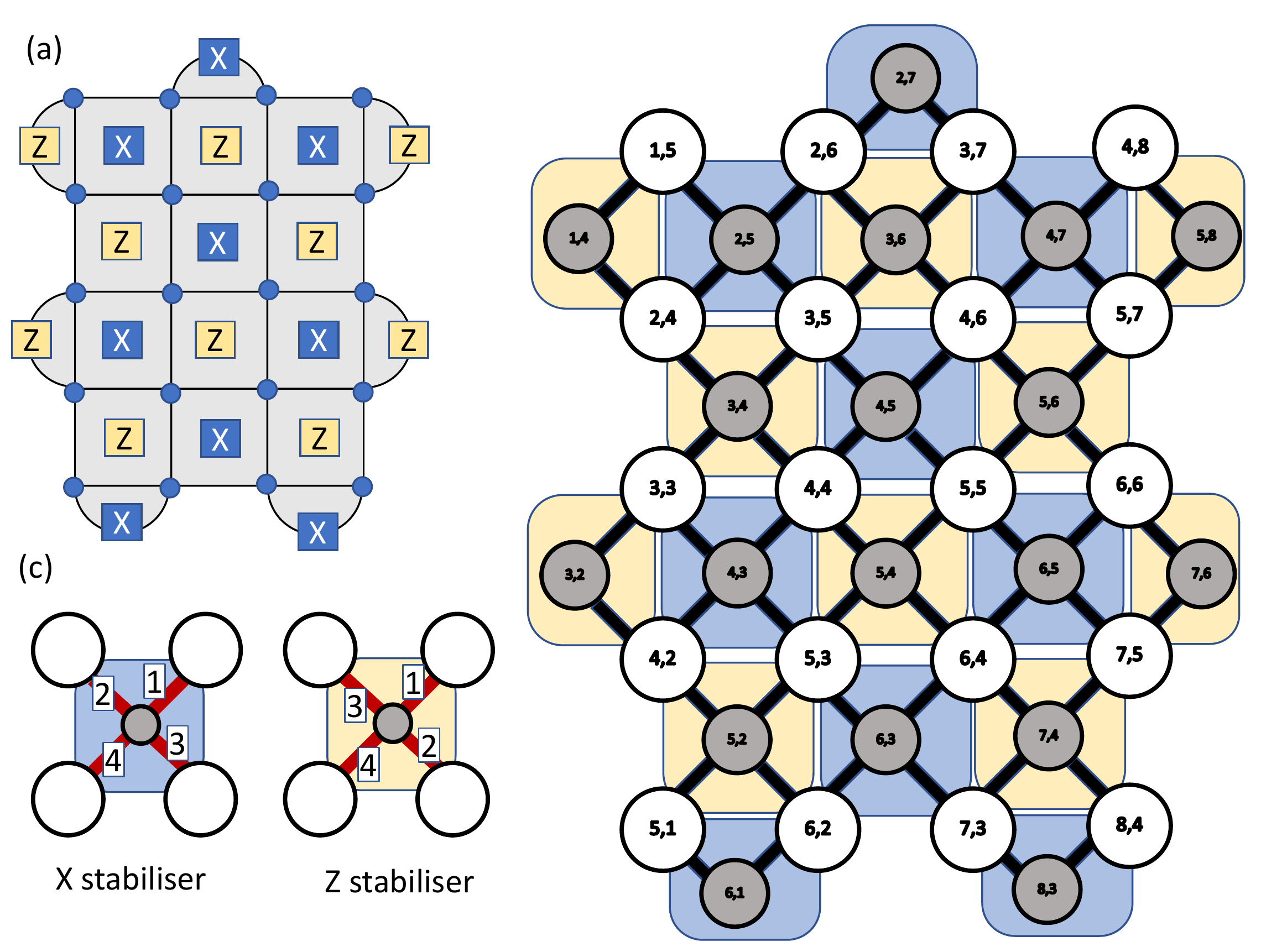}
    \caption{\textbf{(a)} Here we show the standard schematic of rotated surface code. 
    In this case we have a grid of 4x5 data qubits, where the data qubits live on the vertices of the code. 
    The faces represent stabilizer measurements, in this case Z and X stabilizers. 
    In the body of the code the stabilizers are weight four. 
    Boundary conditions are dealt with by smaller weight-two stabilizers as shown. 
    \textbf{(b)} The realisation of the code on the Sycamore device. 
    The numbers in the circles identify the location of the qubit on the device grid. 
    The ancilla qubits (in grey) reside in the centre of each face. The data qubits are shown in white. 
    The black lines represent the connections for two qubit gates that will be utilised to perform the circuits used to prepare the ancillas so they could be measured to perform the stabilizer measurements. 
    \textbf{(c)} In order to minimise the spread of errors the ancillas need to be coupled to the data qubits in a very specific pattern. 
    Here we show the timing pattern of two-qubit gate activation for the ancillas used as Z and X stabilizers. 
    In a complete stabilizer code implementation the ancillas would be measured and reset after the completion of the stabilizer preparation circuits.
    Here we do not do this ancilla measurement (see text). }
    \label{fig:surfacecode}
\end{figure}

\section{Surface code protocol}

The surface code is an error correcting code that appears to be particularly well suited for two-dimensional grids of superconducting qubits. 
\Cref{fig:surfacecode} shows a typical surface code layout and how the abstract code can be mapped to the grid-like qubit devices used by Google (see \cref{sec:devicedetails} for more details of the device). 
In this paper, we consider the standard (CSS) surface code with X- and Z-type stabilizer generators, although we note that more recent experiments on Google devices~\cite{google2022} used the closely related XZZX code instead~\cite{BonillaAtaides2021}. 
Our results could equally be applied in this case. 

Error correction is a technique with many moving parts, all of which have the potential of introducing noise into the system. 
Like any complex system, if all the parts are in action it can be difficult to diagnose the source of problems. 
For this reason, as well as technical limitations of the device, we focus on a simplified analog of error correction circuits. 
Specifically, we run the circuits required for non-demolition four-body stabilizer measurements of the data qubits, but without actually performing the ancilla measurements or the resets required in a real error correction experiment. 
We call these circuits \textit{stabilizer preparation}, since we only prepare to, but don't actually execute the measurement. Here we show how to use efficient protocols introduced in \cite{Flammia2019} and \cite{Harper2019} to extract comprehensive information about the noise in the device while the device is running the stabilizer preparation circuits.

One of the key realisations we utilize is that two rounds of stabilizer preparation of the surface code performs an identity channel on the data qubits. 
This might not be immediately obvious given the complexity of the interleaved two-qubit gates, but is easily verifiable using a Clifford circuit simulator. 
In \cref{sec:protocol} we detail two rounds of stabilizer preparations.

\begin{figure*}[t]
\begin{tikzpicture}[outer/.style={draw=gray,dashed}]
\node[rectangle,anchor=west](step1) at (-11,2) {(a) Step 1: Choose an $m$};
\node[rectangle,anchor=west] at (-11,1.2) {(b) Step 2: Randomise Gates};
\node[rectangle,anchor=west] at (-11,0.4) {(c) Step 3: Apply inverse};
\node[rectangle,anchor=west](step4) at (-11,-0.4) {(d) Step 4: Run circuit};
\node[rectangle,anchor=west](step5) at (-11,-1.2) {(e) Step 5: Repeat $1\dots 4$};
\node[rectangle,anchor=west](step6) at (-11,-2.0) {(f) Step 6: Post-process data};

\node(circ) at (0,0) {\includegraphics[width=0.65\linewidth]{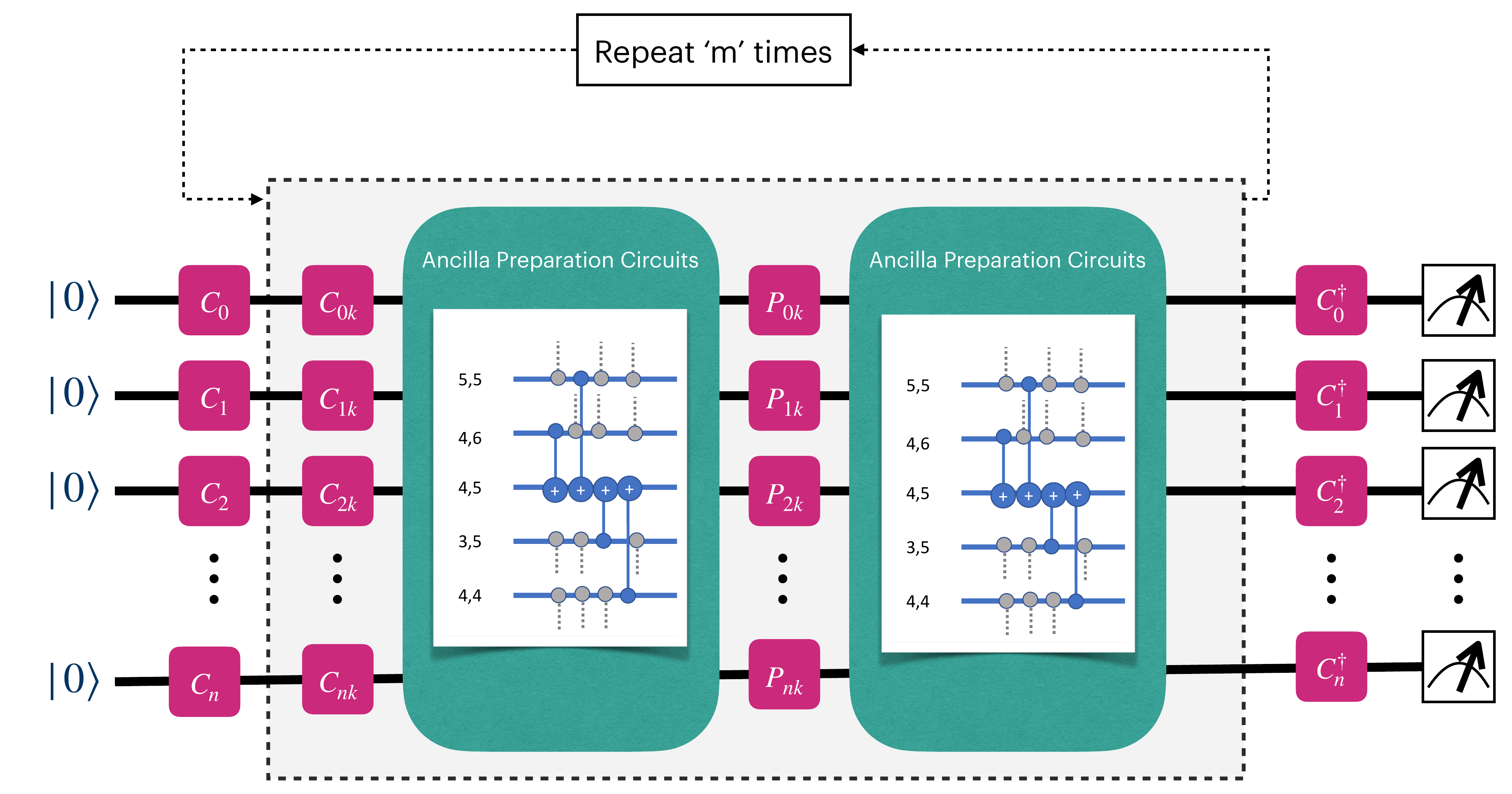}};

\coordinate[left=0.2cm of step5] (Empty2);
\coordinate[left=0.2cm of step1] (Empty1);
\coordinate[right=0.1cm of step4] (Empty3);
\coordinate[left=0.1cm of circ] (Empty4);
\coordinate[below=0.1cm of circ] (belowCirc);
\draw [->] (Empty3) -- +(1,0);
\draw [->] 
  (step5) -- (Empty2) -- (Empty1) |- (step1);

\node(rounds) at (-2.5,-5.5) {\includegraphics[width=0.95\linewidth]{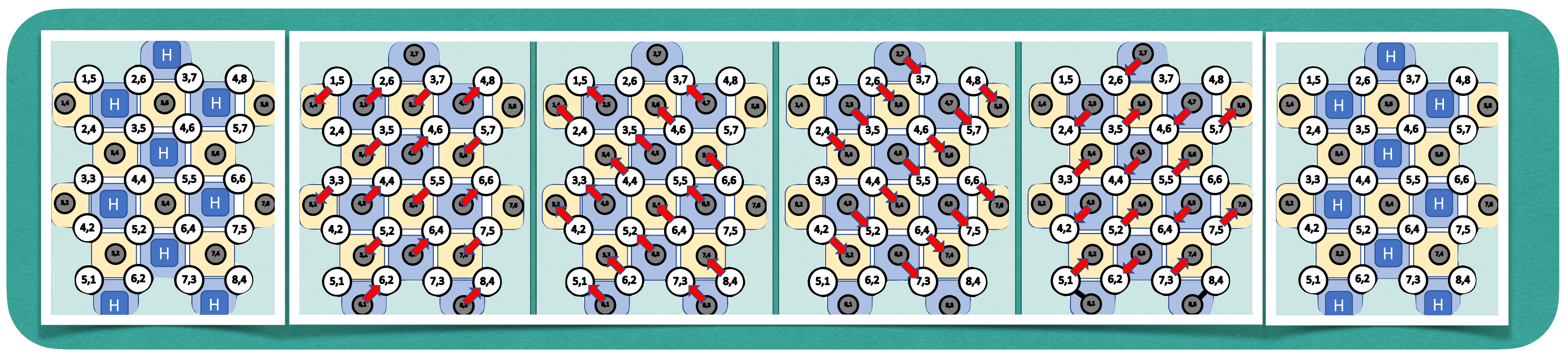}};
\draw[dashed,->] (-1.5,-2.3) -- (-8,-3.5);
\draw[dashed,->] (-1.5,-2.3) -- (3,-3.5);
\draw [->] (-7.5,-7.5) -- (2.5,-7.5);
\node at (-3,-7.7) {Time};
\node[circle,draw,fill=black!20] at (-7.5,-3.9) {1};
\node[circle,draw,fill=black!20] at (-4.9,-3.9) {2};
\node[circle,draw,fill=black!20] at (-2.2,-3.9) {3};
\node[circle,draw,fill=black!20] at (0.4,-3.9) {4};
\end{tikzpicture}

\caption{Schematic of the data acquisition step of the protocol. 
The boxes representing stabilizer preparation circuits are one complete round of stabilizer preparation (without measurement or reset of the ancillas); one such round is detailed in the bottom schematic. 
The numbers (top left of the diagrams in the bottom schematic) represent the order of two qubit gates as shown in \cref{fig:surfacecode}(c), the red arrows representing the two qubit CX gates. 
The large H boxes (left and right sub-plots of the bottom schematic) represent Hadamard gates on the covered qubits.
As with randomized benchmarking the circuits are repeated for multiple sequences over multiple different lengths of circuit. 
The details of the steps are as follows: (a) We randomly choose a circuit length $m$ from 0 (no stabilizer preparation circuits) up to the maximum number of times the block in grey is to be repeated. 
The maximum value for $m$ will depend on the errors in the device, and should heuristically be chosen to be around the inverse of the average error rate of one block. 
(b) For an $n$ qubit system, randomly choose $n$ starting Cliffords ($C_{0\dots n}$) and $n\times m$ random Cliffords and Paulis (respectively, ($C_{0\dots n,k}$), and ($P_{0\dots n,k}$) for k in $\{1\dots m\}$). 
These are used to locally average the noise. 
(c) The random single-qubit Cliffords and Paulis commute through the circuits in a way that is efficiently computable. 
Apply the inverse Clifford required to bring each qubit back into the computational basis. 
Note that we randomly determine whether a qubit is to be returned to the $|0\rangle$ or the $|1\rangle$ state and interpret the results so as to derive the chance of successfully returning to the desired state (i.e.\ record a 0 if the measurement accords with the desired state, or a 1 if it indicates an error). 
(d) Run the circuit for a number of shots (2,000 in our experiment) and record the resultant no-error/error bit-patterns. 
(e) The steps (a)-(d) are repeated multiple times so that the results can be averaged. 
(f) Once the data are gathered we can post-process the results to get the required estimates. 
We describe the data fitting process in mode detail in \cref{sec:protocol}, but in brief, the data can be fit to a model where we fit the data to a sequence of single exponential decays to estimate the relevant eigenvalues for the averaged circuit. 
The fitting to a decay curve makes the estimates independent of state preparation and measurement errors. 
}\label{fig:unitcell}
\end{figure*}
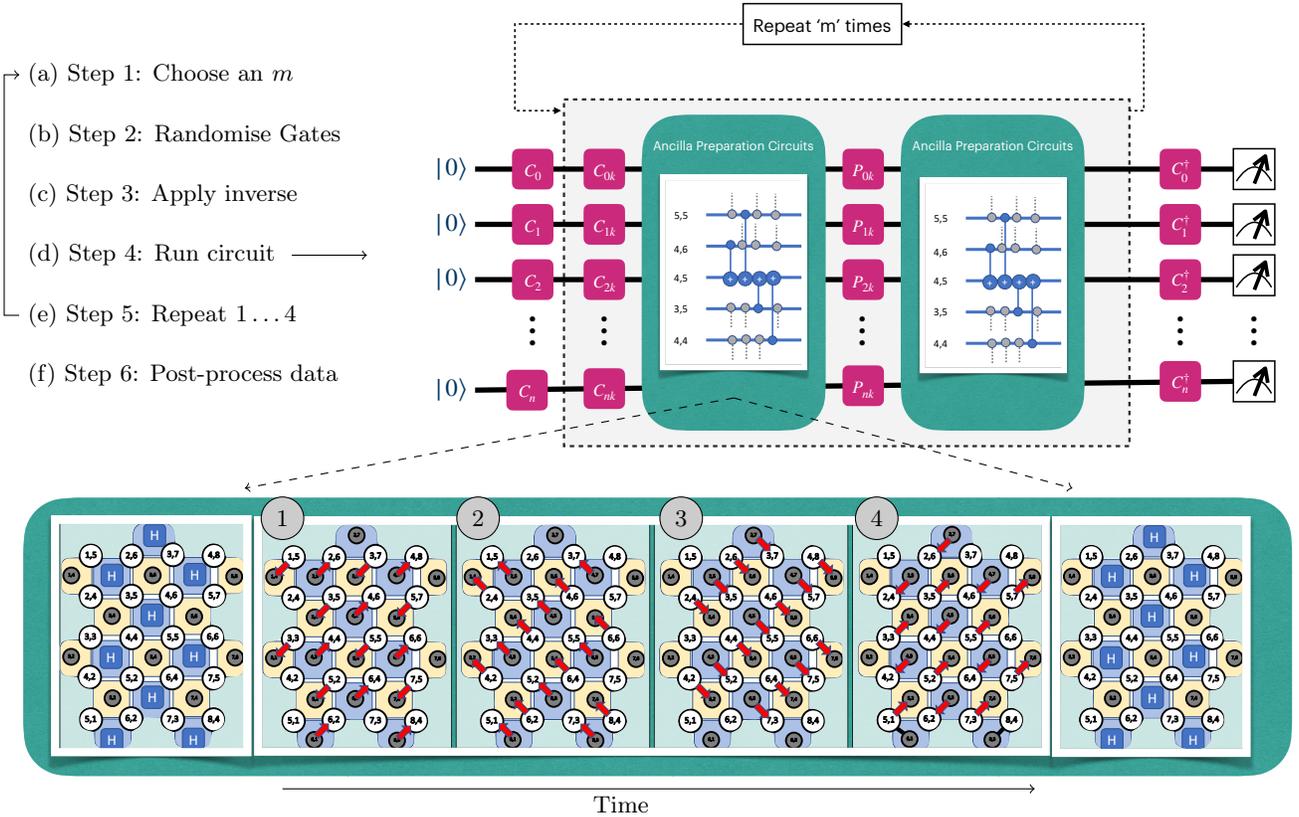

In an actual surface code circuit the ancilla qubits would be measured and reset after each stabilizer preparation round. 
However, where this is not yet possible or, indeed, where one wishes to examine the noise inherent in the circuits without introducing the additional noise that would be caused by measurement and reset, then the circuit extract shown in \cref{sec:protocol} is an example of the circuit required. 
One can also add measurements directly, although some care must be taken in this case (see \cref{sec:protocol} for more details).

To eliminate coherent noise in the circuits, we introduce random Pauli gates into the circuit. 
These gates serve to randomise the Pauli frame, which on average turns the noise into a Pauli channel~\cite{Kern2005,Knill2005,Wallman2016}. 
With Pauli frame randomisation the number of parameters describing the noise is, without further reductions, $4^{d}$ Pauli-channel eigenvalues, where $d$ is the number of data qubits. 

If instead of an initial round of Paulis, we start every two-round block with a round of single-qubit Cliffords, then the analysis in~\cite{Harper2017} applies and we \emph{locally average the noise} in the sense explored in~\cite{Harper2019} and summarised in \cref{sec:locallyaveraged}. 
This local averaging of the noise means we only need to extract $2^{d}$ Pauli-channel eigenvalues, which can be done in the course of this one experiment~\cite{Harper2019}. 
These $2^{d}$ parameters can be learned from extracting only the Pauli $Z$-type eigenvalues of the Pauli noise channel, and all of these eigenvalues can be determined simultaneously since they commute.

With these adjustments in mind, 
the process 
follows the successful core idea of randomized benchmarking~\cite{Emerson2005}: 
After initializing in a product state, we take each two-round block of stabilizer preparation, suitably averaged as above, and repeat it for $m$ steps for varying lengths $m$. 
After an inversion step at the end, we measure in the computational basis. 
The data obtained in this way allow us to separate the state-preparation and measurement errors of the initial round from the errors of the ``unit-cell'' of the two-round block of stabilizer extraction circuits.
An example of this circuit for a single plaquette of the surface code is shown in \cref{fig:unitcell}.
The full description of
all of the steps required
is set out 
in \cref{sec:protocol} and the analysis of the noise estimation is set out in \cref{sec:analysis}.

\section{Experimental implementation}

We tested 
these ideas 
on Google's Sycamore device, a 54-qubit superconducting device of the type described in~\cite{Arute2019}. 
In order to allow correct edge stabilizers, only 39 of the 54 qubits could be utilised, resulting in a $5\times 4$ surface code (i.e.\ 20 data qubits), set out as shown in \cref{fig:surfacecode}(c). 
For the Sycamore device we used sequence lengths $m$ (measured in double rounds of the surface code) of $m = [0, 1, 2, 3, 4, 6, 8]$. 
We ran $1,770$ different sequences, with $2,000$ shots per sequence. 
The total run time was around $8$ hours, in a dedicated $8$ hour time-slot on the device. 

Once we have the data there are a number of things we can immediately do in post-processing. 
The first is simply to marginalise the data to each qubit and fit each such marginalisation as per a standard one-qubit randomised benchmarking experiment. 
This gives us the single qubit (and two-qubit) average error rates to high precision, with a maximum \textit{relative} error of $\pm 0.1\%$. 
Secondly, if all one desired were the two-body correlations, then the data can be marginalised for each pair of qubits and the (normalized) correlations between two qubits ($X$ and $Y$, $\rho_{X,Y}$) calculated as:
\begin{equation}
    \rho_{X,Y}=\frac{E(XY)-E(X)E(Y)}{\sqrt{E(X^2)-E(X)^2}\sqrt{E(Y^2)-E(Y)^2}},\label{eq:moments}
\end{equation}
where $E$ is the expected value of the relevant qubit error rate. 
An example of these data directly extracted from the experiment is shown in~\cref{fig:simpledata}. 
As can be seen in~\cref{fig:simpledata}, there are correlated errors in the device which mostly cluster locally, although some significant longer-range correlations are present. 
While 
the experiment 
is only designed to identify and characterize errors rather than identify their cause, it is possible that the long range two-body correlations noted in the figure and the multi-qubit errors identified by 
 
the experiment
are symptoms of the leakage errors in such Sycamore devices, as more fully explained in~\cite{miao2022}.

What is the impact of nonzero correlation on quantum error correction? 
If we know a particular qubit has an error (say by a syndrome measurement), then the probability of the correlated qubit also having an error changes in a fashion which is computable, in expectation, by this correlation. 
A naive decoder assumes that the error rate on each qubit is independent of whether or not an error has occurred on a different qubit. 
Thus, these correlations can inform a decoder of the changed average error probabilities, improving average performance. 
Including such correlations is in fact necessary to obtain near-optimal performance~\cite{google2022, mcewen2021}. 

However, we can do much more. 
Where the number of qubits is relatively low (no more than, say, 30 data qubits) then it is also possible to extract the entire probability distribution of errors in the device, which includes all larger scale correlations. 
Here the limiting factor is that the probability distribution will have $2^n$ elements and rapidly will become too large to store. 
The details in \cref{sec:protocol} show how to do this.
Since we are interested in characterising even larger scale devices, we need to address this exponential growth in the full probability distribution. 
The experiment itself remains tractable, but if we want to extract the full probability distribution the post-processing scales as $O(2^n)$. 
In principle, knowledge of this full distribution is required for optimal decoding, but the prohibitive complexity makes obvious the need for approximations of this full distribution. 

The rest of the paper explores ways to obtain a useful estimate of the noise without reconstructing the full probability distribution. 
For instance, graphical models of the distribution that represent the noise in a more compact form might give a sufficiently accurate description of the noise for decoding or other purposes, without such severe limits on complexity. 
In the next section we construct such models from our experiments and in \cref{sec:lera} we answer the question of whether such models allow faithful predictions of device performance.

\begin{figure}[t]
\begin{tikzpicture}
\matrix (first) [table,text width=3em]
{
Qubit&Error\\
(1, 5)&0.080\\
(2, 6)&0.152\\
(3, 7)&0.132\\
(4, 8)&0.105\\
(2, 4)&0.087\\
(3, 5)&0.150\\
(4, 6)&0.114\\
(5, 7)&0.102\\
(3, 3)&0.075\\
(4, 4)&0.168\\
(5, 5)&0.182\\
(6, 6)&0.097\\
(4, 2)&0.197\\
(5, 3)&0.215\\
(6, 4)&0.152\\
(7, 5)&0.108\\
(5, 1)&0.230\\
(6, 2)&0.136\\
(7, 3)&0.140\\
(8, 4)&0.092\\
};

\node at (4.3,0) {\includegraphics[width=0.7\linewidth]{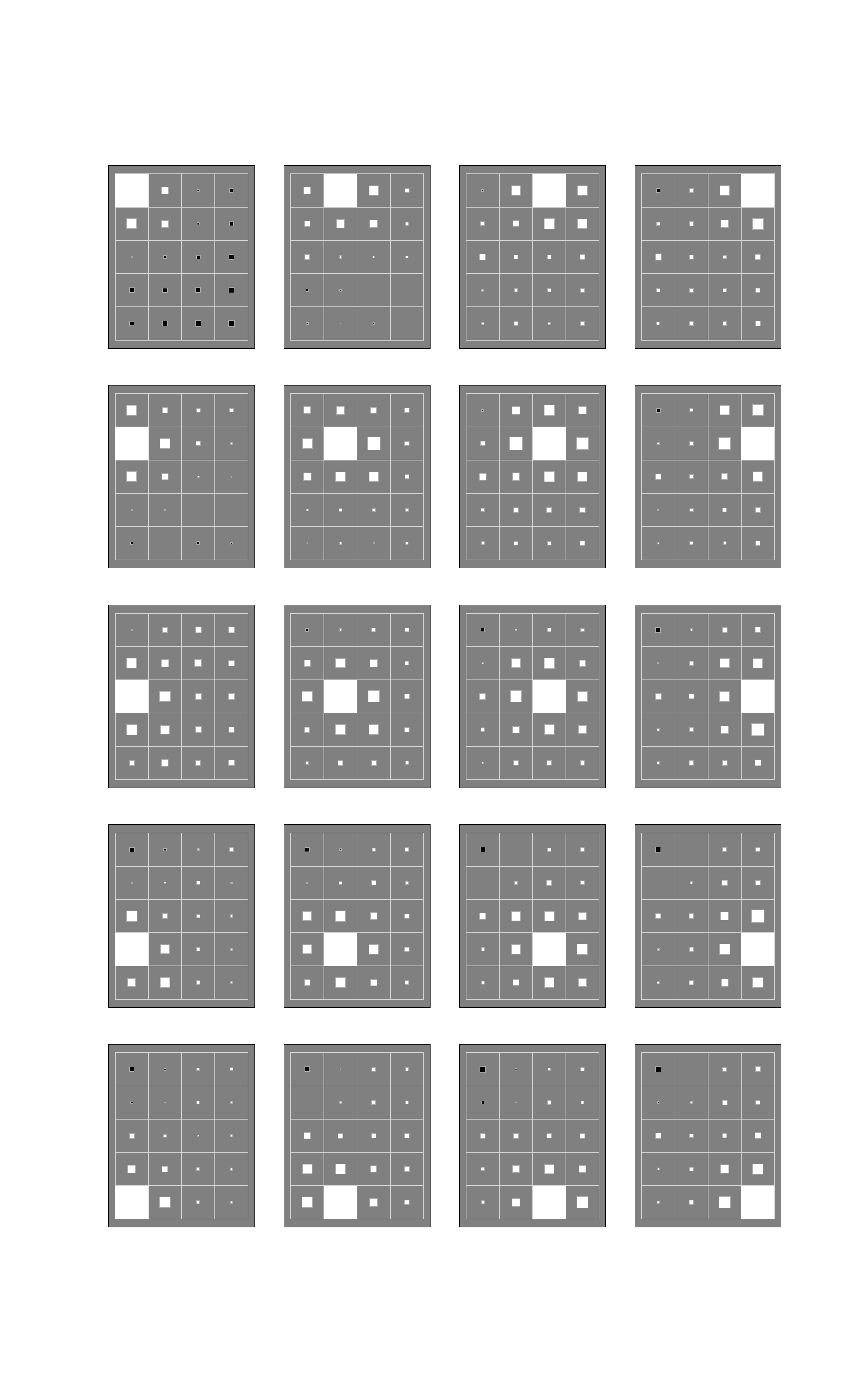}};
\end{tikzpicture}
\caption{Here we show examples of some of the summary data that can be captured in a scalable manner directly from these types of experiments. 
By marginalising the results we can build up a table of the individual error rates, from one round of the surface code (left hand side). 
These error rates represent the probability of any Pauli error occurring on the relevant qubit during one round of the stabilizer preparation circuit.
The right hand side of the figure shows the two-body correlations between each data qubit (solid white in the appropriate sub-graph) and each other data-qubit in the code.
The correlations between the errors on two qubits $X$ and $Y$ ($\rho_{X,Y}$) can be calculated in terms of moments as \cref{eq:moments}. 
This Hinton diagram has a white square (black for negative correlations), where the area of the square is in proportion to the size of the correlation. 
Bootstrapping from the original measured sequences was used to plot two-$\sigma$ error bars, plotted as the width of the line bordering the square (they are barely visible in most cases). 
Of note is that the correlations between data qubits appear (mainly) to be stronger with local data qubits (i.e.\ those close in Manhattan distance according to the device connections).
}\label{fig:simpledata}
\end{figure}

\section{Building and testing noise models}

To go beyond brute force descriptions of global probability distributions, we briefly review the notion of a graphical model, specialized for locally averaged Pauli noise. 
A locally averaged noise model on $n$ qubits can be thought of as a probability distribution on $n$-bit strings, where the presence or absence of an error (either of $X$, $Y$, or $Z$) on the $i$th bit is denoted by a 0 or 1 respectively. 
Under the mild assumption that every error event has nonzero probability, this probability distribution can be described by a Hamiltonian $H(\boldsymbol{x})$ on bitstrings $\boldsymbol{x}$, where the Hamiltonian can be chosen to be $H(\boldsymbol{x}) = \log p(\boldsymbol{x})$, the log of the probability of $\boldsymbol{x}$. 

In this picture, an ``energy shift'' by a constant factor corresponds to changing the normalization of the probability distribution, which can often be neglected for sampling or modeling purposes. 
There is no need to invoke the physical concept of temperature since this is a formal mapping. 
We are interested in modeling a probability distribution as a graphical model whose Hamiltonian has only a bounded number of interactions per qubit, each among a bounded number of other qubits, in order to keep learning tractable. 

The simplest example of this mapping is where a single nonidentity Pauli error occurs with probability $p$ and otherwise no error occurs. 
Then $H(0) = \log(1-p)$ and $H(1) = \log(p)$. 
This Hamiltonian can be written as $H(x) = f + h x$, where $h = \log(p/(1-p))$ and where $f = \log(1-p)$ controls the (often optional) normalization. 
For $n$ qubits and general independent noise, one would have the Hamiltonian $H(\boldsymbol{x}) = \sum_k h_k x_k$ (here and henceforth we drop normalizations). 
The special case of identically distributed noise occurs when $h_k = h$ for all $k$.
To generate interesting correlations, we need to add \textit{coupling terms} to the Hamiltonian, and one is naturally lead to models with Hamiltonians of the form
\begin{equation}
    H(\boldsymbol{x}) = \sum_k h_k x_k + \sum_{j,k} J_{j,k} x_j x_k + \ldots\,.
\end{equation}
The simplest nontrivial case might contain only nearest-neighbor two-body correlations, for example. 

For these models, learning the entire probability distribution is equivalent to learning the couplings of the associated Hamiltonian. 
Examples of such models are given in~\cref{fig:markovBlanket}.
In \cref{fig:markovBlanket}(a) we assume independent and identically distributed (IID) noise only (so that $H(\boldsymbol{x}) = h \sum_k x_k$). 
In \cref{fig:markovBlanket}(b) we keep independence, but relax to non-identically distributed (IND) noise, corresponding to $H(\boldsymbol{x}) = \sum_k h_k x_k$. 
\Cref{fig:markovBlanket}(c) shows the simplest model of nearest-neighbor correlation, an Ising-style model with $H(\boldsymbol{x}) = \sum_k h_k x_k + \sum_{\langle j,k\rangle} J_{j,k} x_j x_k$, where the second sum is over nearest neighbors. 
Finally, we consider a model that coarse-grains the device into a 1D array. 
This allows arbitrary correlations along one row of the device, but limits vertical correlations to those between qubits in adjacent rows; this is shown in~\cref{fig:markovBlanket}(d). 
The associated coarse-grained 1D (CG1D) model Hamiltonian has up to 8-body coupling terms.

The choice of which which graphical models to examine is well motivated. 
The `Ising style' model is relevant as the long-range spread of errors from one data qubit to another should be limited in the device, the gates coupling the data qubits to the ancillas acting as a type of `one-way' gate for $Z$ and $X$ errors~\cite{Versluis2017}. 
Indeed, in \cref{sec:ce} we show how the data confirm that (for this device) the Ising model is a good choice compared with other potential decompositions of equivalent expressive power. 
Similarly motivated, the CG1D model is also able to capture longer correlations, albeit at the cost of exponential scaling in the width of one dimension. 
While problematic for larger systems, it might still be a valid methodology in systems with highly biased noise, where the proposed grid would be rectangular rather than square~\cite{chamberland2022,napp2013}.

While these models attempt to match the correlation structure in the device, some model error is inevitable. 
For instance, the Ising model assumes local interactions (which we imagine arising from gate errors), but there are many other noise mechanisms that can cause errors to spread from one qubit to another. 
Examples include energy state leakage~\cite{miao2022}, leakage of control signals and qubit frequency crowding. 
Model error in this case will indicate such processes at work.

One method of measuring model error is to use the relative entropy, or more flexibly, the smoothed symmetrized version of it, the Jensen-Shannon divergence (JSD)~\cite{Lin1991}. 
One can measure the JSD as between the global distribution and the distribution encapsulated in the model to quantify the model error. 
Here we have all the data required to make this calculation. 
However, to do so in general requires the full distribution, and therefore this is not a scalable solution. 
Another potential method might be to form the covariance matrices embodied in each of the distributions --- this can be done in a scalable manner. 
(Obviously models (a) and (b) in~\cref{fig:markovBlanket} have no inter-qubit correlations.) 
These covariance matrices only capture the two-body correlations in the distributions and technically only form a lower bound on the total variational distance (TVD) of the distributions (\cref{sec:covbound}), but are usable as a guide in most non-pathological cases. 
We set out these calculations in \cref{fig:metrics}.
Unsurprisingly, as can be seen from \cref{fig:metrics}, the richer the model the more complexity we can capture in the correlations between the qubits.

The question then arises not as to whether these more sophisticated models can capture characteristics of the noise that cannot be captured by the simpler models, but whether those additional characteristics are important in determining and/or predicting error correction properties in the device. 
We now address this question, focusing on the logical error rate of the device.

\begin{table}
\begin{tabular}{p{1.2cm}p{0.3cm}p{1.5cm}p{0.3cm}p{1.5cm}}
\toprule
Model                         & \multicolumn{2}{c}{JSD($D\lVert M$)} & \multicolumn{2}{c}{$\|\Sigma_D-\Sigma_M\|$} \\ \midrule
\multicolumn{1}{l|}{IID}      && 0.229                  && {0.124}                      \\
\multicolumn{1}{l|}{IND}    && 0.192                  && {0.090}                    \\
\multicolumn{1}{l|}{Ising} && 0.167                  && {0.056}                    \\
\multicolumn{1}{l|}{CG1D}    && 0.148                  && {0.019}                  \\ \bottomrule
\end{tabular}
\caption{Metrics for the observed error distribution ($D$) compared with the distributions embodied in the various models ($M$). 
The two error metrics are the norm between the respective correlation matrices and the Jensen-Shannon distance (JSD) (see main text for definitions). 
The models are IID (identical, independently distributed errors) \cref{fig:markovBlanket}(a), IND (independent non-identically distributed errors) \cref{fig:markovBlanket}(b), Ising, with nearest-neighbor two-qubit correlations only \cref{fig:markovBlanket}(c), and the coarse-grained 1D (CG1D) distribution shown in \cref{fig:markovBlanket}(d).}\label{fig:metrics}
\end{table}

\begin{figure}[t!]
    \includegraphics[width=\columnwidth]{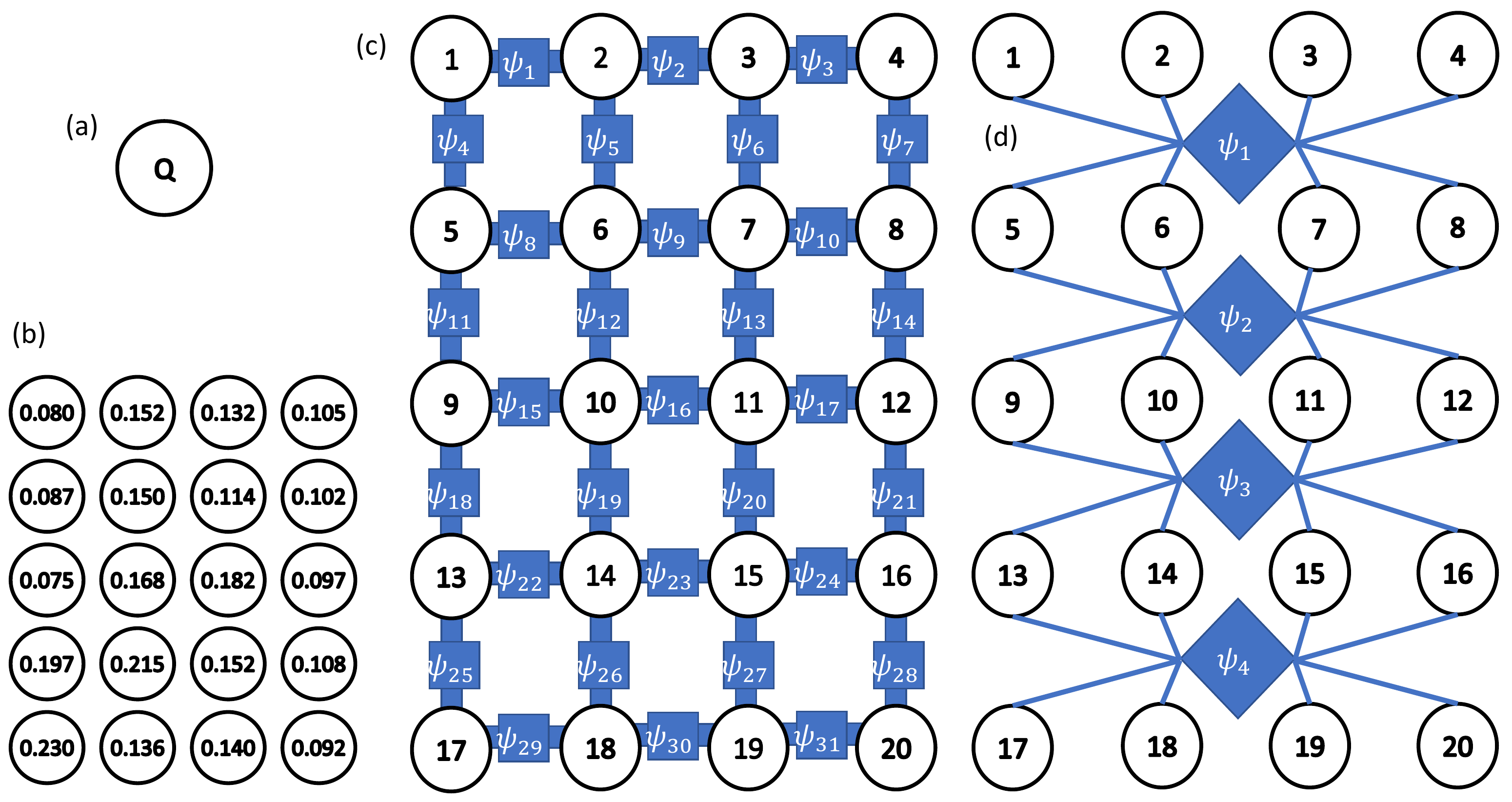}
    \caption{Various graphical models used for modeling the observed error distribution. 
    Sites connected by ``factors'' $\psi_i$ can be correlated by arbitrary Hamiltonian couplings supported adjacent to the factors. 
    \textbf{(a)} A typical model used for decoder testing. 
    All the qubits are assumed to have identical independent depolarizing noise, so there is only one node. 
    \textbf{(b)} Here the qubits are still independent but each has their own error rate (here the error rate shown in~\cref{fig:simpledata}). 
    This model requires $n$ parameters, i.e.\ one for each data qubit --- in this case 20. 
    \textbf{(c)} Here we model the noise with an Ising model where the only factors $\psi_i$ are between neighboring qubits. 
    \textbf{(d)} This model is referred to in the text as the Coarse-Grained 1D (CG1D) model, where we have reduced the surface code to a one-dimensional graph. 
    }
    \label{fig:markovBlanket}
\end{figure}

\section{Logical Error Rate Analysis}\label{sec:lera}

\begin{figure*}[t!]
    \includegraphics[width=\textwidth]{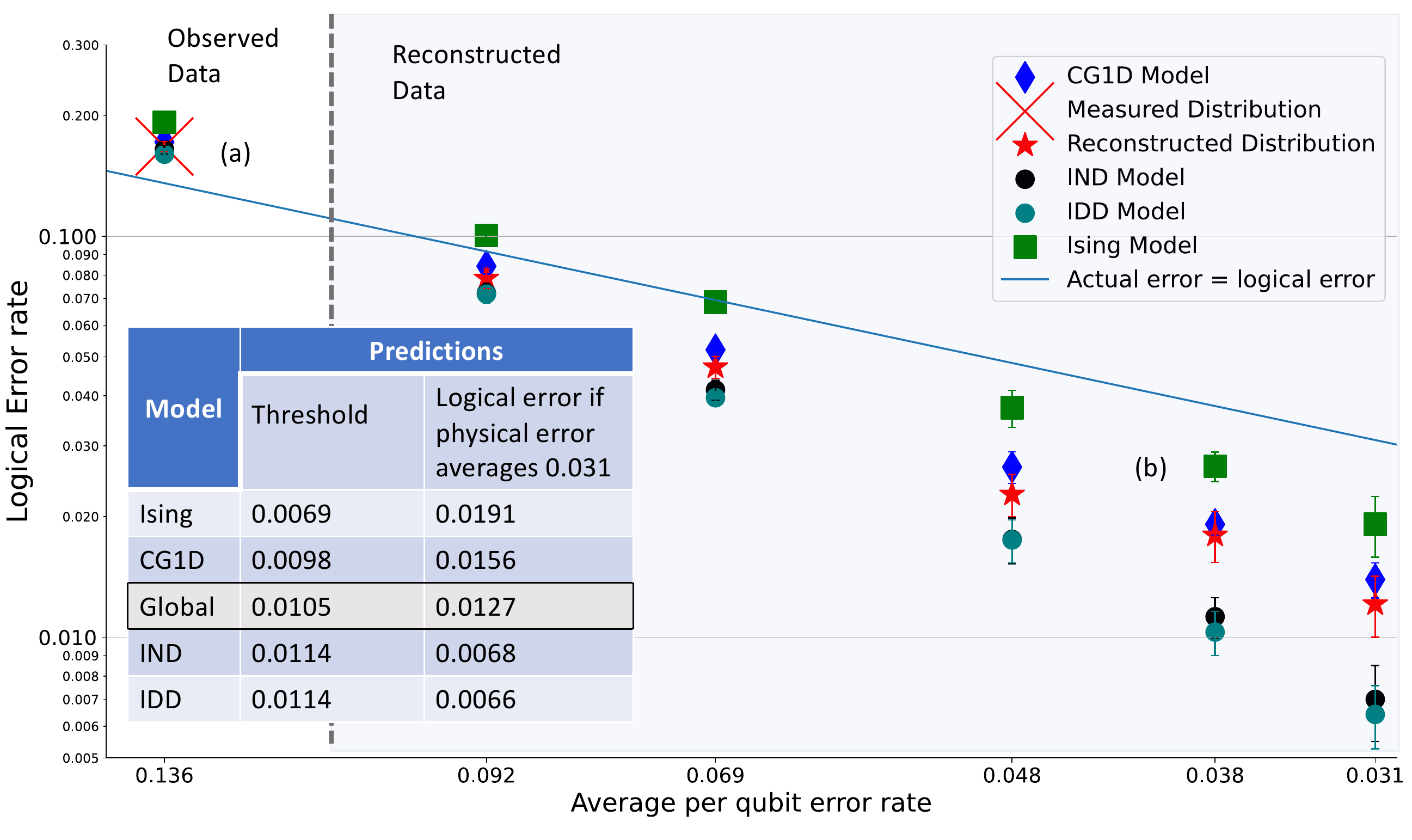}
    \caption{
    Here we plot the logical error rate of various physical error distributions, where we have used a generic  QECSIM~\cite{qecsim} decoder (a tensor network decoder with bond dimension $\chi=8$) to decode the error syndromes. The process involved: 1) Measuring the empirical global distribution using \acr, as discussed in this manuscript. 2) Constructing counterfactual global error distributions with smaller error rates (\emph{reconstructed distributions}). The methodology for this step is discussed in the text and \cref{sec:theoretical}. 3) For each of these distributions (observed and reconstructed), each model of interest (see~\cref{fig:markovBlanket}) is constructed. 4) The relevant distributions and each of the models are then sampled to provide a sample of the errors given by each of them. For each sample we use the QECSIM decoder to determine whether a successful decoding of the error syndrome would be made. (To enable the sampled locally averaged errors to simulate a full Pauli distribution we assume an X,Y or Z Pauli error with  equal probability for every affected site.) 5)~Sufficient samples are taken of each distribution and each model of that distribution to allow the estimation of a logical error rate for that distribution/model of the distribution. 6) Steps 4 and 5 are repeated 10 times to generate reported bootstrap error bars.}
    \label{fig:errorPlot}
\end{figure*}
In order to generate counterfactual error distributions that are less noisy but still contain similar correlation structures, we consider the following one-parameter family of error channels associated to a probability distribution $p$ on bit strings. 
The error probabilities $p$ can be related to the eigenvalues $\lambda$ of the superoperator representation of the channel using $W$, the Walsh-Hadamard transformation~\cite{Flammia2019,Harper2019,berg2021}.
In fact, we have $\lambda = W p$, and this map is invertible to obtain $p = W^{-1}\lambda$. 
If this superoperator with eigenvalues $\lambda$ were generated by a continuous time process, then we can generate the family of channels with eigenvalues 
\begin{equation}
    \lambda(t) = \exp(t \log \lambda(1))\,,\quad \lambda(1) = W p\,.\label{eq:lt}
\end{equation}
When $t=0$, this is the identity channel, and $t=1$ gives back the original noise channel eigenvalues. 
We can study the error rates of this distribution by considering 
\begin{equation}
    p(t) = W^{-1} \lambda(t)\label{eq:pt}
\end{equation}
for various values of $t\ge 0$, with smaller values corresponding to less noise, but with a qualitatively similar correlation structure as the true channel $p(1)$. 
Note that larger integer values $t=2, 3, \ldots$ correspond to multiple applications of noise. 

There is an important subtlety: not every input $p$ is ``divisible'' into a continuous time process, even in the simple case of Pauli noise channels. 
We therefore add an extra step of projecting $p(t)$ defined by \cref{eq:pt} to the nearest point on the probability simplex. 
See \cref{sec:theoretical} for more discussion and analysis of the probability distributions obtained by this method.

To summarize: our experiment gives us the observed distribution of the actual noise in the device, and our graphical models have given us estimates of this noise in a convenient and scalable format. 
By making the additional assumption that the noise is generated by a continuous time process, we use \cref{eq:lt} and \cref{eq:pt} to generate counterfactual theoretical noise distributions which are \emph{estimates} of the noise that would exist in the device if the noise channel were `less' noisy, but retained the correlation features of the observed noise channel. 
We can also construct models of this counterfactual noise.

The following calculations of the logical error rate are based on two additional important assumptions. 
First, as mentioned previously, this experiment omits the measurement and reset of ancillas required for true quantum error correction. 
These will, undoubtedly, introduce complexity, noise, measurement errors and timing issues in any final circuit. 
The numbers obtained and discussed below must be regarded with this caveat in mind. 
Second, by using a generic decoder, we are not attempting to utilise our knowledge of the noise to \emph{improve} the decoding process. 
The decoder utilised is a generic decoder provided by QECSIM~\cite{qecsim}. 
While there is much work yet to be done to work out if actual knowledge of the noise can improve decoding success rates \cite{tiurev2022}, this is not what we do here --- rather the analysis below might be seen as setting out the base success rate, which such decoders might seek to improve on. 
Writing such decoders is the subject of on-going work. 

With this in mind, we can now calculate the logical error rate of both the observed and constructed noise channels, and we can determine if the more sophisticated models of that noise are better able to predict the likely logical error rate than their primitive counterparts. 
If so, then it is likely they are capturing essential elements of the noise that the simpler models cannot. 

\Cref{fig:errorPlot} shows the results of these logical error rate estimates for the different noise models. 
In each case 10,000 error samples were taken to estimate the logical error rate, repeated 10 times to provide the error bars. 
Of interest in the plot is the behaviour when the average per-qubit physical error rate is approximately 0.136 (point (a) in the plot). 
This is the measured noise in the device.
At that point the logical error rate is larger than this (0.176), but notably there is very little difference in the logical error rate between the various models (from the most primitive IID model to the CG1D distribution) and the full observed distribution. 
At this point in the spectrum of noise channels the simple IID noise model is as good as any (point (a) in the figure). 

However, as we use higher fidelity reconstructed distributions a clear difference emerges. 
By the time we hit point (b) in the figure, the simpler models result in an estimated logical error rate that is significantly lower than the logical error rate predicted by the full reconstructed distribution. 
The correlated errors have an impact in this regime, curtailing the generic decoder's ability to correct for errors. 
(Again we note that a decoder tailored for these noise characteristics might perform much better.) 
This accords with beliefs that correlated noise in a machine will have an impact on logical error rates. 
Interestingly even in this regime the CG1D model is within error bars of predicting the same logical error rate as the full reconstructed distribution. 

The Ising model is pessimistic in its predictions (predicting a higher logical error rate). 
Our belief is that it is due to the Ising model not being able to correctly capture longer range correlations that happen to exist in this device.
The parameterization of the Ising model therefore incorporates these longer range correlations into the short range correlations it can model, resulting in a model that has stronger short range correlations than the full distribution.

In this small (low-distance) surface code implementation, it is these stronger short range correlations that have the largest impact on the logical error rate.
Despite this limitation as we show in \cref{sec:ce}, if we constrain the model to use only a limited number of two-qubit factors, the Ising model is still the optimal model of such distributions, minimising model error (and possibly avoiding overfitting). 
On these data one could be confident that if an Ising model of the distribution, populated from data from the device, was below threshold, the actual device is also likely to be below threshold.
\bigskip

\section {Conclusions}
We have presented and implemented a method of characterising an important part of the surface code, namely the stabilizer preparation portion of an error correction circuit. 
We have shown how one can use random benchmarking-style experiments to measure the locally averaged noise and the Pauli noise on the data qubits used in the code. 
We have shown how to create graphical models of the noise which continue to be tractable as surface code sizes increase.
These models allow us to explore important questions related to the ability to run error correction on the device. 
Finally we presented empirical evidence that shows how these more sophisticated models appear to be increasingly necessary if one wishes to accurately predict the errors in the device as error rates get lower and correlations shorter range. 
The importance of taking such errors into account is highlighted by the data we extrapolate from the device, providing support to the belief that decoders that take into account the actual noise of the device could potentially lead to higher threshold implementations of error correcting codes.

\section{Code and data availability}
All code and data is available upon reasonable request.

\section{Acknowledgments}

The experiment was performed in collaboration with the Google Quantum AI hardware team under the direction of Y.\ Chen, J.\ Kelly, and A.\ Megrant. 
We acknowledge the work of the team in fabricating and packaging the processor; building and outfitting the cryogenic and control systems; executing baseline calibrations; optimizing processor performance; and providing the tools to execute the experiment. 
We thank D.\ Debroy, B.\ Foxen, M.\ Harrigan, and M.\ Newman for thoroughly reading the manuscript and providing helpful feedback. 
RH would like to thank Robin Blume-Kohout for several insightful conversations relating to the paper.
RH is funded by the Sydney Quantum Academy and this work was supported by ARO grant W911NF2110001.
SF's contributions to this project were completed while he was affiliated with the University of Sydney. 
\bibliography{refs}

\newpage

\appendix

\section{Noise Characterization}\label{sec:background}

Many characterization techniques have been developed to understand noise in quantum devices. 
Traditionally full-characterization techniques such as process tomography have been used when devices are in their infancy. 
For small devices the qubits can be fully characterised with a view to gaining insight in to the underlying causes of noise. 
Process tomography \cite{Chuang1997} and its variants \cite{Merkel2012,Blume-Kohout2016,evan2022} remain an important tool in the characterization arsenal. 
However, such methods are limited in that they can only be applied to a small number of qubits. 
Where we have more than a few qubits, scalablity issues render the techniques impractical.

Strong noise-averaging or partial characterization is a way to overcome the scalability issues of full characterization. 
A popular partial characterization technique is randomized benchmarking and its variants~\cite{Emerson2007, Knill2008, Magesan2011, Carignan-Dugas2015, Hashagen2018, Erhard2019, Helsen2019, Proctor2018, proctor21b, hines2022}.
These techniques define a natural measure of a `block' of noise over which the noise in the device is averaged.
By defining a block of noise that is well behaved through multiple applications, simple decay curves can used to estimate parameters in a way that eliminates state preparation and measurement (SPAM) errors and provides small error bars~\cite{Harper2018}. 
Importantly, these protocols average the noise in the device, and this multiple-run average turns the noise into the equivalent of a Pauli channel removing coherent effects of the noise. 
Depending on the protocol, the Pauli noise itself may be additionally averaged, often reducing the noise to a single value, the fidelity of the overall channel. 
The conversion to a Pauli channel is justified as the noise in the device can itself be converted into a Pauli channel during the operation of a device by using techniques such as Pauli frame randomization~\cite{Kern2005,Knill2005,Wallman2016}. 
The subsequent strong-averaging is a characterization convenience.

Rather than seeking to measure the average noise of the entire Pauli channel, further advances have shown how to harness the power of simultaneous single qubit measurements to extract much more information about the noise \cite{Flammia2019, Harper2019, huang2020, cotler20, Harper2021, flammia2021, flammia2021b}.
Recently it has been shown one can estimate Pauli noise channels from error correction syndrome measurements~\cite{wagner2021}, although in the case of surface-code style error correction this is limited to a maximum of two-body correlations, irrespective of the size of the code. 
Here we will use techniques that allow the extraction of the global error distribution \cite{Flammia2019,Harper2019} on devices running error correcting circuits and show how to tame this exponentially large amount of data by constructing appropriate models of the noise.

\section{Device Characteristics}\label{sec:devicedetails}

\begin{figure*}[th!]
    \centering
    \begin{tikzpicture}
	\node[inner sep=0pt] at (-0,6.5)         {\includegraphics[width=1\textwidth]{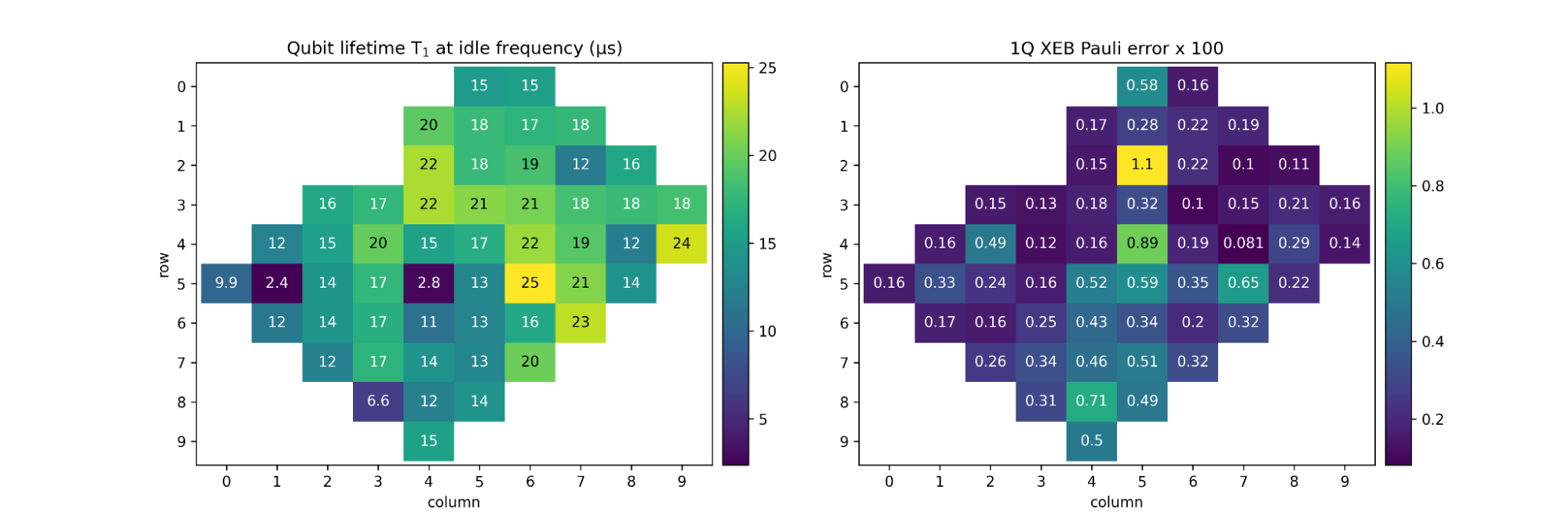}};
	\node at (-7,9.5) {\textbf{(a)}};
	
	\node[inner sep=0pt] at (0,0)         {\includegraphics[width=1\textwidth]{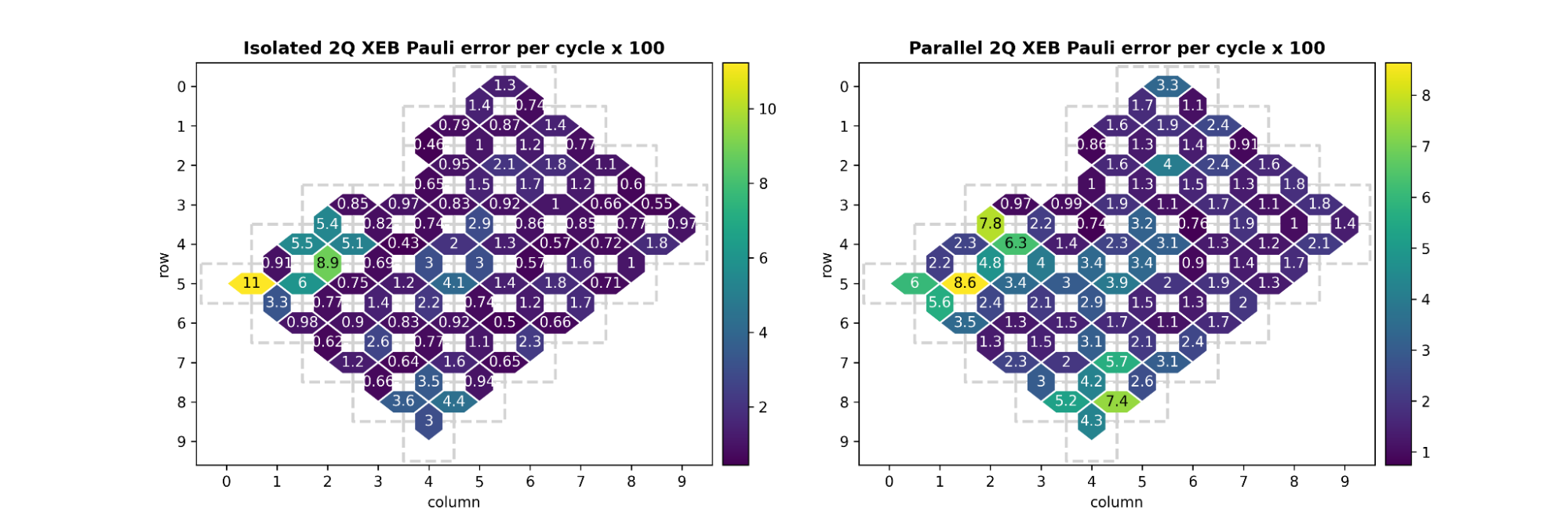}};
	\node at (-7,2.8) {\textbf{(b)}};
	
	\node[inner sep=0pt] at (0,-6.5)         {\includegraphics[width=1\textwidth]{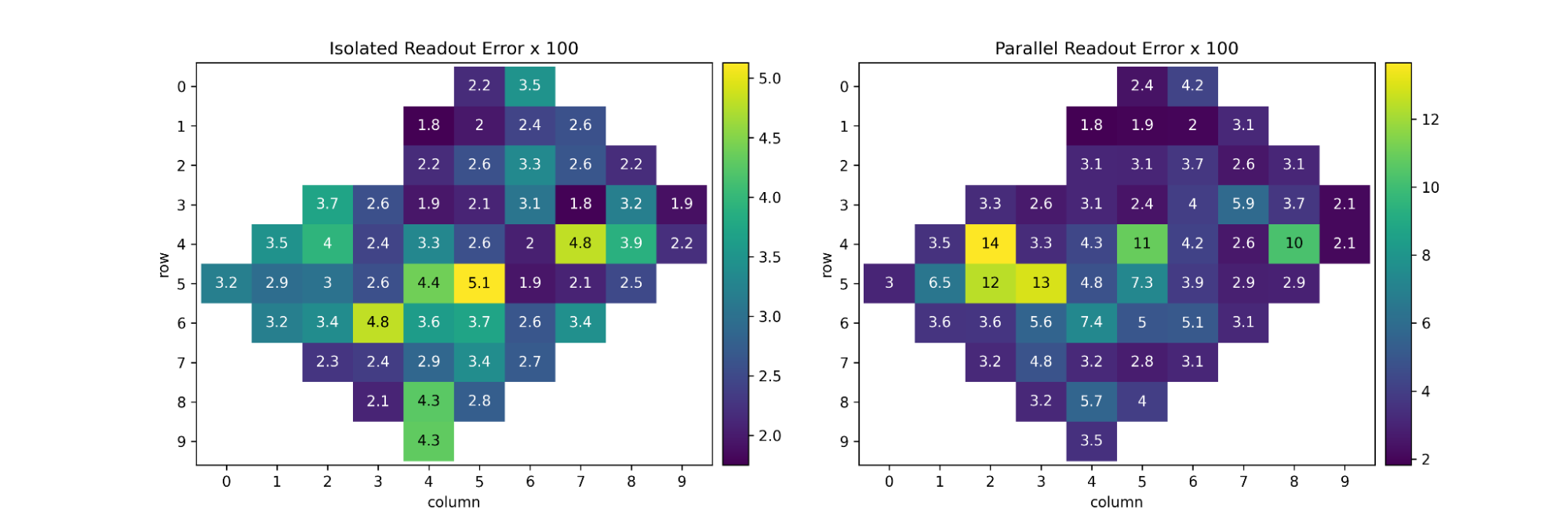}};
	\node at (-7,-3.6) {\textbf{(c)}};

	\end{tikzpicture}	
	
    \caption{\textbf{(a)} typical qubit T1 lifetimes and their frequency and single qubit Pauli errors. \textbf{(b)} isolated and parallel two qubit XEB Pauli errors. \textbf{(c)} isolated and simultaneous measurements errors. In all cases see \cref{sec:devicedetails} for further details.}
    \label{fig:google}
\end{figure*}
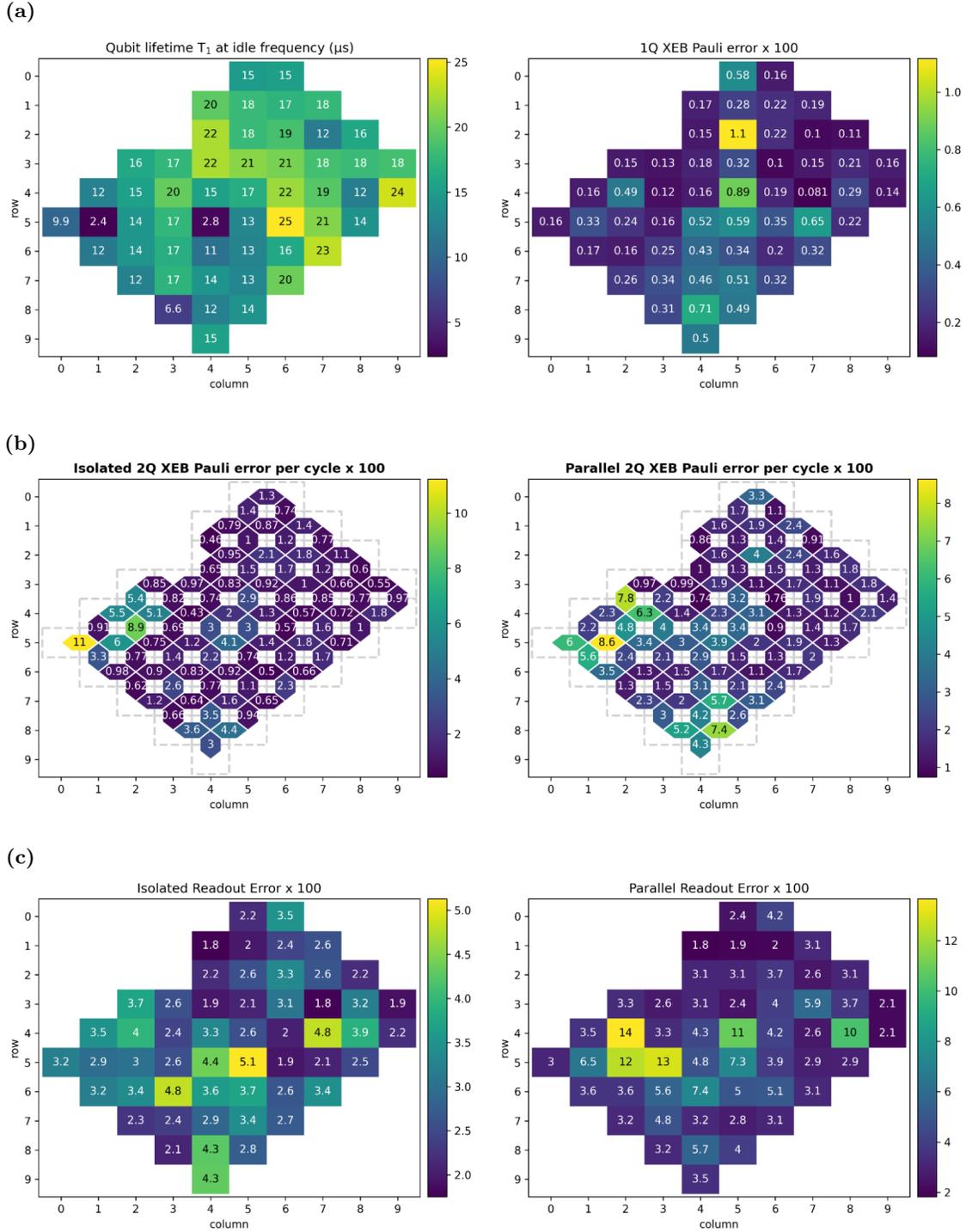

The device used in these experiments is the exactly the same 54 qubit sycamore processor as in Ref.~\cite{Arute2019}. 
We operate 39 qubits selected from a 53 qubit grid with idle and interaction frequencies chosen to be optimized for on-resonance sqrt(iSWAP) and Sycamore gates; while running circuits on these 39 qubits, the unused qubits in the 53 qubit grid idle at their normal operating frequencies and are not biased to low frequencies. 
Automated calibrations were performed and control parameters updated just three times a week with minimal manual intervention. 
Instabilities in coherence times and electronics drift were not compensated for in the intervening time. 
Most data presented in this paper were acquired approximately 60 hours after the previous calibration and control parameter update.

Immediately following each automated calibration cycle, gate and readout fidelities were characterized to provide representation device performance characteristics as summarized in \cref{fig:google}. 
\Cref{fig:google}(a) shows typical qubit T1 lifetimes at their idle frequencies, and \cref{fig:google}(b) plots single qubit Pauli errors for $\pi/2$ gates characterized with isolated cross entropy benchmarking (XEB), where isolated operation is defined as applying gates to a single qubit with the remaining 52 idling at their idle frequency. 
On this processor we have observed only a modest increase in single-qubit gate errors during simultaneous operation (see fig.~2 of Ref.~\cite{Arute2019}) but significant differences are present for isolated and simultaneous readout and two-qubit gate operations. 
Since the precise mechanisms of crosstalk are not well understood, we assume that performance is dependent on exactly which set of qubits are operated simultaneously. 
To provide a representative estimate of these crosstalk effects, in \cref{fig:google}(b) and \cref{fig:google}(c) we provide isolated and simultaneous measurements of Sycamore gate fidelity and readout error. 
Note: Sycamore gate fidelity is characterized in four layers to characterize gates using each of the (up to) four couplers connected to each qubit. 
Finally we note again that \acr{} is robust to SPAM errors, although high SPAM errors will reduce the signal and therefore increase the number of measurements needed to achieve the same accuracy.

While \cref{sec:protocol} describes the circuits run in terms of Clifford gates, natively the Sycamore device runs Sycamore gates for its two-qubit gates. 
The CX gates are translated into two rounds of Sycamore gates
\begin{align*}
    \text{SYC}=\left[\begin{array}{cccc}1&0&0&0\\0&0&-i&0\\0&-i&0&0\\0&0&0&e^{-i\pi/6}\end{array}\right]
\end{align*}
with appropriate single qubit rotations. The single qubit rotations (being the random single qubit Cliffords, the Sycamore gate corrections, the random Paulis and the required Hadamard gates can also be combined into single Phased XZ gates, defined as
\begin{align*}
    &\text{pXZ}(x,z,a)=\\&\left[\begin{array}{cc} \mathrm{e}^{i\pi x/2}\cos(\pi x/2) & -i\mathrm{e}^{i\pi(x/2-a)}\sin(\pi x/2)\\-i\mathrm{e}^{i\pi(x/2+z+a)}\sin(\pi x/2) & \mathrm{e}^{i\pi(x/2+z)}\cos(\pi x/2)\end{array}\right]
\end{align*}
This is done automatically by using the appropriate Cirq functions to optimise the circuits for the Sycamore device.
 
 An  extract of the circuit for a particular selection of random Cliffords and Paulis is shown as \cref{fig:circuitextract}(b)

\section{Description of the experiments}\label{sec:protocol}

As with most SPAM-robust protocols that measure incoherent noise rates, the experiment design utilizes many of the features of randomized benchmarking. 
We use single-qubit Clifford gates to locally average the noise, with additional Pauli frame randomization where we can't easily use Clifford gates to remove coherent noise between iterations of the surface code circuits.

\begin{figure*}[ht!]
    \centering
    \begin{tikzpicture}
	\node[inner sep=0pt] at (-0,6.5)         {\includegraphics[width=\textwidth]{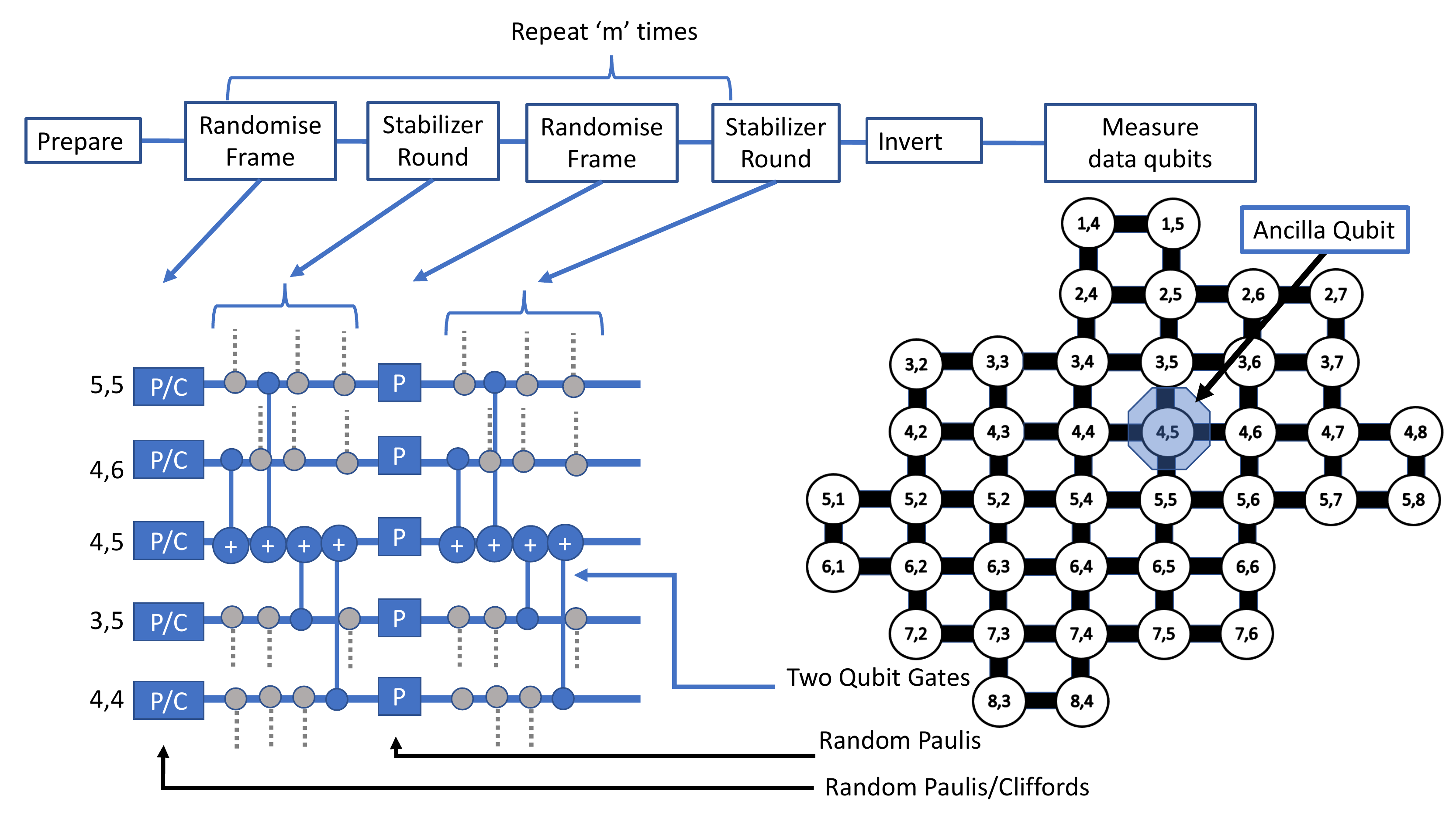}};
	\node at (-8.5,11) {\textbf{(a)}};	
	\node[inner sep=0pt] at (0,-1)         {\includegraphics[width=\textwidth]{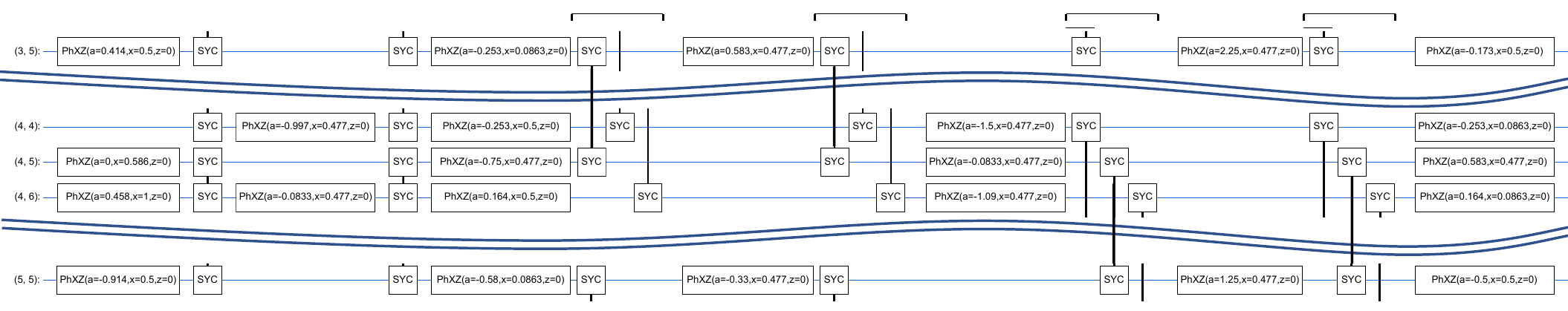}};
	\node at (-8.5,1) {\textbf{(b)}};
    \end{tikzpicture}	
    
    \caption{(a) Here we show an example of the circuits we measure. 
    These are the same circuits used to prepare the ancilla qubits, so as to allow syndrome extraction.
    We show only the gates for one ancilla (qubit 4,5). 
    (We have included `shadows' of other gates just to remind the reader that many other two qubit gates are being executed simultaneously with the gates shown; these `shadow' gates are in gray). 
    The inclusion of random Pauli gates ensures Pauli frame randomization, which when averaged over many runs means that the statistics gathered are as if the noise channel we are measuring were a Pauli noise channel. 
    If we replace the first round of Pauli gates with random single qubit Clifford gates, then we can locally average the noise in a sense made more precise in the text. 
    The circuit extract shown here is a round of gates.
    Each round returns the qubits to the computational basis, subject to a random single qubit Pauli (or Clifford) on the qubits. 
    This is easily tracked and inverted just prior to readout. 
    By preparing and measuring sequences for circuits with varying lengths (i.e.\ vary $m$ in the graphic above), we are able to transform observed probabilities into eigenvalues and fit to a decay curve (see \cite{Flammia2019,Harper2019}). 
    Ref.\ \cite{Flammia2019} proves the convergence properties of such circuits to estimate the probability distribution of the average noise in the system. 
    In this case it is the average noise while running exactly the type of circuit we are interested in. 
    (b) The native two qubit gates on the device are the Sycamore gates. 
    Two Sycamore gates (and some single qubit gates) are required to implement a CX. 
    Here we show most of one round of the above circuit extract (with some randomly chosen single qubit Clifford gates) detailing the gates actually executed on the device. 
    Where a Sycamore gate connects to a qubit not shown in this extract, only a single leg is shown. }
    \label{fig:circuitextract}
\end{figure*}

An overview of the design and the related circuit extract is shown as \cref{fig:circuitextract}. The full procedure is as follows:

\begin{enumerate}
\item Choose a non-negative integer $m$.

\item For each qubit randomly decide whether to leave the state invariant or to map it to an orthogonal state in order to eliminate a nuisance model parameter (see e.g.\ Refs.~\cite{Fogarty2015,Andrews2019,Harper2018}).

\item Create $m$ applications of the following circuit:
\begin{itemize}
    \item random single qubit Clifford gate on each qubit.
    \item apply one round of the stabilizer preparation circuits.
    \item apply a random Pauli gate on each qubit.
    \item apply one round of the stabilizer preparation circuits.
\end{itemize}

Note: if $m$ is 0, then still apply the single qubit Clifford gate and the inverting gate in step 4.

\item Choose the random single qubit Clifford gate needed to return each particular qubit to the state chosen in Step 2.

\item Run the circuit created in 2--4 a number of times (in the experiment we chose 2,000), measuring all the qubits. 
Record the bit-patterns from the measurements.

\item Repeat steps 1--5 for various values of $m$, in the current experiment we chose $m \in [0, 1, 2, 3, 4, 6, 8]$. 
Ref.~\cite{Harper2018} provides guidance but in general you want the largest $m$ to to be large enough that most qubits have a marginalised survival rate that is $\approx 60\%$.

\item For each distinct $m$ chosen repeat steps 1--6 sufficient times to obtain reasonable statistics (in the reported experiment we ran each $m$ over 1,700 times, but far fewer runs are actually required for reasonable error bars).

\end{enumerate}

If the number of data qubits ($n$) is such that the $2^n$ is a tractable number:
\begin{enumerate}
    \item For each $m$, marginalise the bit-patterns to the data qubits and use these reduced bit patterns to create a $2^n$ outcome empirical probability distribution. 
    It is not a concern that not all possible bit patterns will have been observed. 
    These values are 0 in the empirical probability distribution.
    \item Walsh-Hadamard transform each of these $m$ probability distributions, forming $m$ eigenvalue vectors (each eigenvalue vector having $2^n$ entries).
    \item For each of the $m$ eigenvalue vectors, for $i$ in $2\dots2^n$, extract the $i^{\text{th}}$ eigenvalue ($\bar{\lambda_i}$) and fit to the following equation:
    \begin{align}\label{eq:decay_model}
        \bar{\lambda_i}(m) = A f_i^m 
    \end{align}
    Record each $f_i$ from the fitting procedure.
    \item Form a new eigenvalue vector $[1,f_2 \dots f_{2^n}]$. 
    If desired, Walsh-Hadamard transform this to a probability distribution and project onto the nearest point in the probability simplex.
\end{enumerate}

If $2^n$ is not a tractable number, marginalise the bit-patterns to tractable chunks and use the resultant marginal probability distributions to determine the parameters of a chosen model. 
Connecting the (empirical) marginal probabilities to the model parameters of an undirected graphical model to the model parameters is discussed at length in \cite[Sec.\ VII]{Flammia2019} and we describe how to do this parameter estimation for our models in more detail in \Cref{sec:factors}.

Measurements of the ancilla qubits can be included by expanding the circuits where it is desired to include ancilla measurement and reset. 
Some care is required, however, as such measurement and reset must only be done after two rounds of stabilizer preparation (unlike actual execution of the surface code, where measurement and reset occurs after each round of stabilizer preparation). 
Performing the experiment on a circuit without measurement and reset, and then with measurement and reset will allow the noise associated with the measurement and reset to be extracted in a similar way to interleaved benchmarking.

\section{Analysis of the data}\label{sec:analysis}
The main idea behind applying the analysis to stabilizer preparation circuits is that we can leverage the self inverting nature of a round of these circuits to implement the protocol described in \cite{Harper2017}. 
As stated in that paper the protocol is applicable to any `gate' A where $A^2$ and $APA^\dagger$ are elements of the Clifford group (here P is a Pauli). 
In this case a round of stabilizer preparation circuits, which can be thought of as a large multi-qubit gate (being composed of Clifford gates) satisfies those conditions. 
In more detail, let $\sccmacro$ represent a round of stabilizer preparation circuits, and $\Lambda_\sccmacro$ represent the noise on such a round. 
We use $\sim$, as in $\widetilde{\sccmacro}$, to represent a noisy `gate'; here a noisy implementation of a round of stabilizer preparation circuits. 
Given this we have
\begin{align}
    &\widetilde\sccmacro = \sccmacro\Lambda_\sccmacro\\
    &\sccmacro\sccmacro = I \implies \sccmacro = \sccmacro^\dagger\label{eq:invert}
\end{align}
where, in the first line, we have arbitrarily written the noise on the on the right hand side of the gate. (Nothing depends on this.)

We use $C_j$ to represent a series of $n$ single qubit gates, the suffix representing the $j^\text{th}$ draw of a set of single qubit Cliffords and $n$ being the number of qubits used to implement a round of the surface code (both data and ancilla qubits). 
$\mathcal{C}$ represents the set of all possible instantiations of $C$ and $|\mathcal{C}|$ is the size of this set. Similarly $P_j$ represents $n$ single qubit Paulis and $\mathcal{P}$ and $|\mathcal{P}|$ have analogous meanings.

In this notation a round of the protocol looks like:
\begin{equation}
    \bra{0}^{\otimes n} C_\text{inv} \dots C_{j+1} \tilde{\sccmacro} P_{j+1} \tilde{\sccmacro} C_j \tilde{\sccmacro} P_j \tilde{\sccmacro}\dots \ket{0}^{\otimes n}
\end{equation}

Three further things to note, 
1) We have ignored the noise on the rounds of Paulis and single qubit Cliffords, they will be trivial compared to a full round of the surface code/stabilizer preparation circuits and where possible the transpiler will compile them into the physically realized gates, 
2) Paulis commute through a round of the stabilizer preparation circuits into other Paulis (since it is comprised of Clifford gates), i.e.\ 
\begin{equation}
    P\sccmacro = \sccmacro P'
\end{equation} 
and 3) the sequence $P \sccmacro$, being a round of perfect Paulis followed by a perfect round of stabilizer preparation circuits, is a unitary-1 design. 
This means that conjugating a noise channel with $P \sccmacro$, averaged over all Paulis will remove the coherence in the noise, reducing it to a Pauli channel, i.e.\ 
\begin{align}
\frac{1}{|\mathcal{P}|}\sum\limits_{P\in\mathcal{P}}P\ \tilde{\sccmacro}\ P'\ \sccmacro = 
& = \frac{1}{|\mathcal{P}|}\sum\limits_{P\in\mathcal{P}}P\ \sccmacro\ \Lambda_\sccmacro\ \sccmacro\ P \label{eq:commuteP}\\
& = \frac{1}{|\mathcal{P}|}\sum\limits_{P\in\mathcal{P}}(P\ \sccmacro)\ \Lambda_\sccmacro(P\ \sccmacro)^\dagger\label{eq:invertSC} \\
& = \Lambda_{P\sccmacro}
\end{align} where in \cref{eq:commuteP} $P'$ has been chosen so that it commutes through $\sccmacro$ to become $P$ and in \cref{eq:invertSC} we have used the fact that for both the Paulis and for a round of the stabilizer preparation circuits they are their own inverse, \cref{eq:invert}. Here $\Lambda_{P\sccmacro}$, represents the Pauli-twirled noise channel on a single round of the surface code, being a Pauli channel.

With this is hand we can now analyse the proposed protocol. 
The repeating portion of the protocol is as follows:
\begin{align}
    C_{j+1} \tilde{\sccmacro}\ P_{j+1}\ \tilde{\sccmacro}\ C_j 
    =\ &C_{j+1}\ \sccmacro\ \Lambda_\sccmacro P_{j+1} \sccmacro\ \Lambda_\sccmacro\ C_j
\end{align}
Now we can let $P'$ be defined as $P_{j+1} \sccmacro = \sccmacro P'$ and define $C'_{j+1}$ as $C_{j+1} = C'_{j+1} P'$, noting that if $C_{j+1}$ was chosen randomly, $C'_{j+1}$ is also a random Clifford. 
Then
\begin{align}
    C_{j+1}\ \sccmacro\ & \Lambda_\sccmacro P_{j+1} \sccmacro\ \Lambda_\sccmacro\ C_j \\
    & = C'_{j+1}\ P'\ \sccmacro\ \Lambda_\sccmacro P_{j+1} \sccmacro\ \Lambda_\sccmacro\ C_j\\
    & = C'_{j+1}\ (P' \sccmacro)\ \Lambda_\sccmacro\  (\sccmacro P')\ \Lambda_\sccmacro\ C_j.
\end{align}
Finally we can write $C'_{j+1}$ as $C''_{j+1} C_j^\dagger$ 
\begin{align}
    &C''_{j+1} (C_{j})^\dagger\bigl[\ (P' \sccmacro)\ \Lambda_\sccmacro (\sccmacro P') \Lambda_\sccmacro\ \bigr]C_j,
\end{align}
which gives us a Pauli twirl of the noise, embedded in a Clifford twirl of the noise (when averaged over sufficient sequences). 
The exact nature of this twirl is analysed in detail in \cite{Harper2017}. 
A question might be asked as to why not just use rounds of single qubit Cliffords to locally average the noise, rather than the alternating Clifford and Pauli rounds analysed above. 
The reason is that while single qubit Cliffords can indeed be commuted through an odd number of rounds of the surface code, the resulting Clifford is no longer in the group of simultaneous single qubit Cliffords, but rather is a number of multi-qubit Clifford gates. 
Chasing this through the circuits results in a multi-qubit inversion gate at the end, which will substantially reduce the signal being measured.

\section{Locally averaged noise channel}\label{sec:locallyaveraged}

The meaning of a locally averaged noise channel is  fully explored in~\cite{Harper2019}, which builds on the proofs in~\cite{Flammia2019}. 
However, for the sake of completeness we summarise the results here. 
It is well known that averaging observed distributions over sequences of gates drawn from particular groups eliminates coherence and will average the noise in a way dependent on the groups in question. 
This is colloquially known as `twirling' the noise.
Ref.~\cite{Helsen2019} explores this in detail. 
If the group in question is the Pauli group, then the noise being twirled, when averaged over repeated randomised sequences, is the equivalent to a Pauli noise channel, i.e.\ the noise with all of its coherence removed. 
When the group is the full Clifford group, then the noise becomes a quantum depolarizing channel with the same fidelity as the original channel~\cite{Nielsen2002}, this is one of the essential ingredients of randomized benchmarking~\cite{Magesan2011}. 
Twirling the noise with single qubit Clifford gates, averages the noise, such that the number of distinct eigenvalues of an $n$-qubit noise channel are reduced from the $4^n$ distinct eigenvalues of a Pauli channel, to $2^n$ distinct values. 
It does this by averaging the Pauli noise locally. 
For example, in a two qubit channel, the Paulis $IX$, $IY$ and $IZ$ are averaged, the Paulis $XI$, $YI$ and $ZI$ are averaged and the remaining 9 two-qubit Paulis are averaged. 
In this case the 4 distinct eigenvalues ($2^2$) can be recovered by measuring the eigenvalues of the $II$, $IZ$, $ZI$ and $ZZ$ Paulis. 
This can be done by observing the appropriate measurements (which form a probability distribution over 4 outcomes) and applying a Walsh-Hadamard transform on that probability distribution. 
The Walsh-Hadamard transform is a form of Fourier transform that moves from a probability distribution (which may be sparse) to the dense eigenvalue basis of the channel and vice-versa (again see \cite{Flammia2019} for details).

\section{Parameter estimation for the factors}\label{sec:factors}

\begin{figure}[t!]
    \includegraphics[width=1\columnwidth]{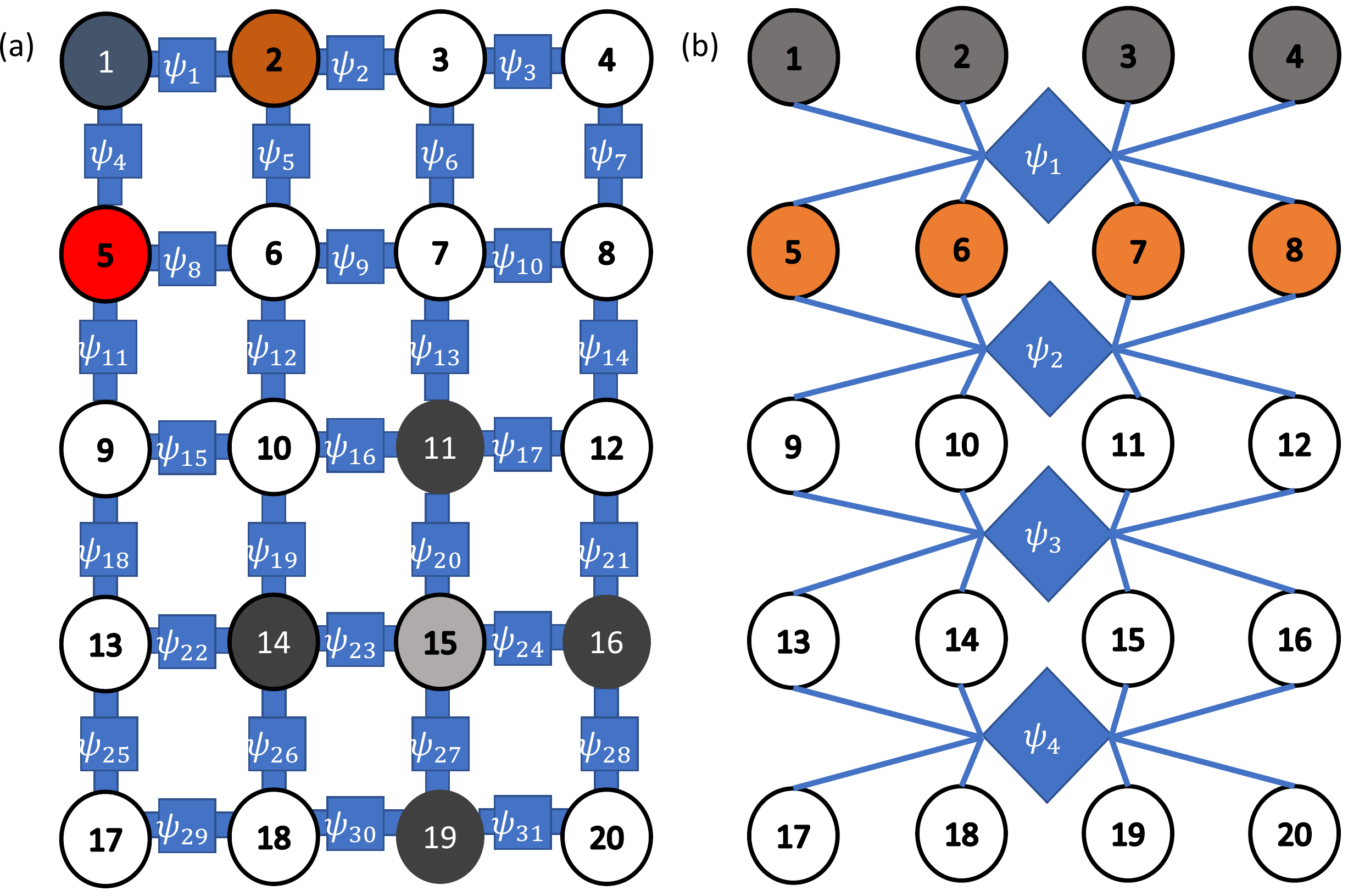}
    \caption{Here we illustrate the Markov blanket assumption built into the Ising model factor graph and the CG1D graph. As discussed in the text the errors on each qubit are independent of all other qubits if one is given the qubits that share a factor with the qubit in question. The variables are the data qubits and are associated with the circular ``variable'' nodes in the graphs, labelled 1 through 20. The square (diagonal) nodes represent the ``factor'' nodes. In the graphs we have color coded some of the qubits to aid with identification in the dicussion below. (a) 
    Qubit 1 (dark grey) has two factors ($\psi_1$ and $\psi_3$) that touch it. 
    The only other qubits these factors touch are qubits 2 and 5 (red). 
    Therefore the assumption is that the errors on qubit 1 are independent of the other qubits (excluding qubits 2 and 5) given the errors on qubits 2 and 5. 
    Similarly the Markov blanket on qubit 15 (light grey) is formed from qubits 11, 14, 16 and 19 (black). Meaning that the assumption built into the model is that the distribution on qubit 15 is independent of the other qubits, if one knows the errors on qubits 11, 14, 16 and 19.
    There are 31 factors required in this graph, each factor specifies a $2^2$ distribution, meaning the model requires 124 parameters. 
    (b) This model is referred to in the text as the Coarse-Grained 1D (CG1D) model, where we have reduced the surface code to a one-dimensional graph. 
    In this case the independence assumptions can be read from top to bottom and we have that qubits 1,2,3 and 4 (light grey) are independent of qubits 9 through to 20, given qubits 5, 6, 7 and 8 (orange). 
    There are only four factors in this model, but each factor has a distribution of $2^8$, meaning that this model requires 1024 parameters to be fully specified. 
    This is the only model with exponential scaling; it scales exponentially in the width of the grid. 
    We note that the choice between width and height factorization is arbitrary, in this case we chose width because the factors would be smaller. 
    In practice it made little difference to the results presented.
    }
    \label{fig:markovBlanket2}
\end{figure}

Here we give the intuition behind the scalable creation of a Markov network such as the Ising model using the locally averaged noise distribution. 
As mentioned previously for a device with $n$ data qubits, this is just a classical probability distribution of size $2^n$. 
Consider such a distribution over $n$ bits $\boldsymbol{x} = x_1x_2\dots x_n$, which we can write as $p(x)$. 
We are looking to efficiently represent this distribution. 
One well known approach is to model it as an undirected graphical model or Gibbs random field~\cite{Hammersley1971}. 
Here the probability distribution is associated to variables that live on the ``variable'' nodes of a factor graph, and the interactions live on the ``factor'' nodes that describe how to couple the variables, as we now detail. 

The underlying assumption behind such a graphical model is the \emph{global Markov} property, which is contained in the structure of each particular factor graph $G$. 
In such a graph (e.g.\ (a) of \cref{fig:markovBlanket2}), we associate each variable $x_j$ with a data qubit, which  we also label $x_j$. 
The ``factor'' nodes connect these variables, and are labelled by the factors, here $\psi_i$. 
These describe the correlations that are possible in $p$ according to the \emph{global Markov property}. 

To define the global Markov property, let us assume that there is a separating set $S$ of vertices with variables $\boldsymbol{x}_S$ between two sets of variables $\boldsymbol{x}_A$ and $\boldsymbol{x}_B$ (i.e.\ every path in $G$ between $A$ and $B$ goes through at least one vertex in $S$ ) then the global Markov property asserts:
\begin{equation}
    p(\boldsymbol{x}_A,\boldsymbol{x}_B|\boldsymbol{x}_S) = p(\boldsymbol{x}_A|\boldsymbol{x}_S)p(\boldsymbol{x}_B|\boldsymbol{x}_S)
\end{equation}
for every $A$ and $B$ and any separating set $S$. 
That is, that the marginal distribution over $x_A$ and $x_B$ is conditionally independent give the values on the separating subset $\boldsymbol{x}_S$. 
We illustrate this with examples in \cref{fig:markovBlanket2}. 

The Hammersley-Clifford theorem \cite{Hammersley1971} says that every strictly positive probability distribution that obeys the global Markov property for a factor graph $G$ factorises over the factor nodes of $G$, such as the graph shown in (a) of \cref{fig:markovBlanket2}. 
That is, $p(\boldsymbol{x}) = \tfrac{1}{Z}\prod_k \psi_k$, 
where $\psi_k$ are the factor functions, $Z$ is a normalization and the factors $\psi_k$ for each factor node are positive functions. 

Because the functions $\psi_k$ are positive, they can be (and often are) reparameterized as $\psi_k = \exp(H_k)$ where $H_k$ is a real-valued function (an abstract ``Hamiltonian'') that is a function only of the variables in a neighborhood of at most one factor. 
In that case, the normalizing factor $Z$ takes on the interpretation of the partition function of the Gibbs distribution of the total Hamiltonian $H = \sum_k H_k$ at an inverse temperature $\beta=1$. 
This physics language is helpful, but we gently remind the Ising model aficionados that the variables here are bits, $x \in \{0,1\}$, so some intuitions coming from spin variables can be misleading in this context.

\begin{figure}[t!]
    \includegraphics[width=1.0\columnwidth]{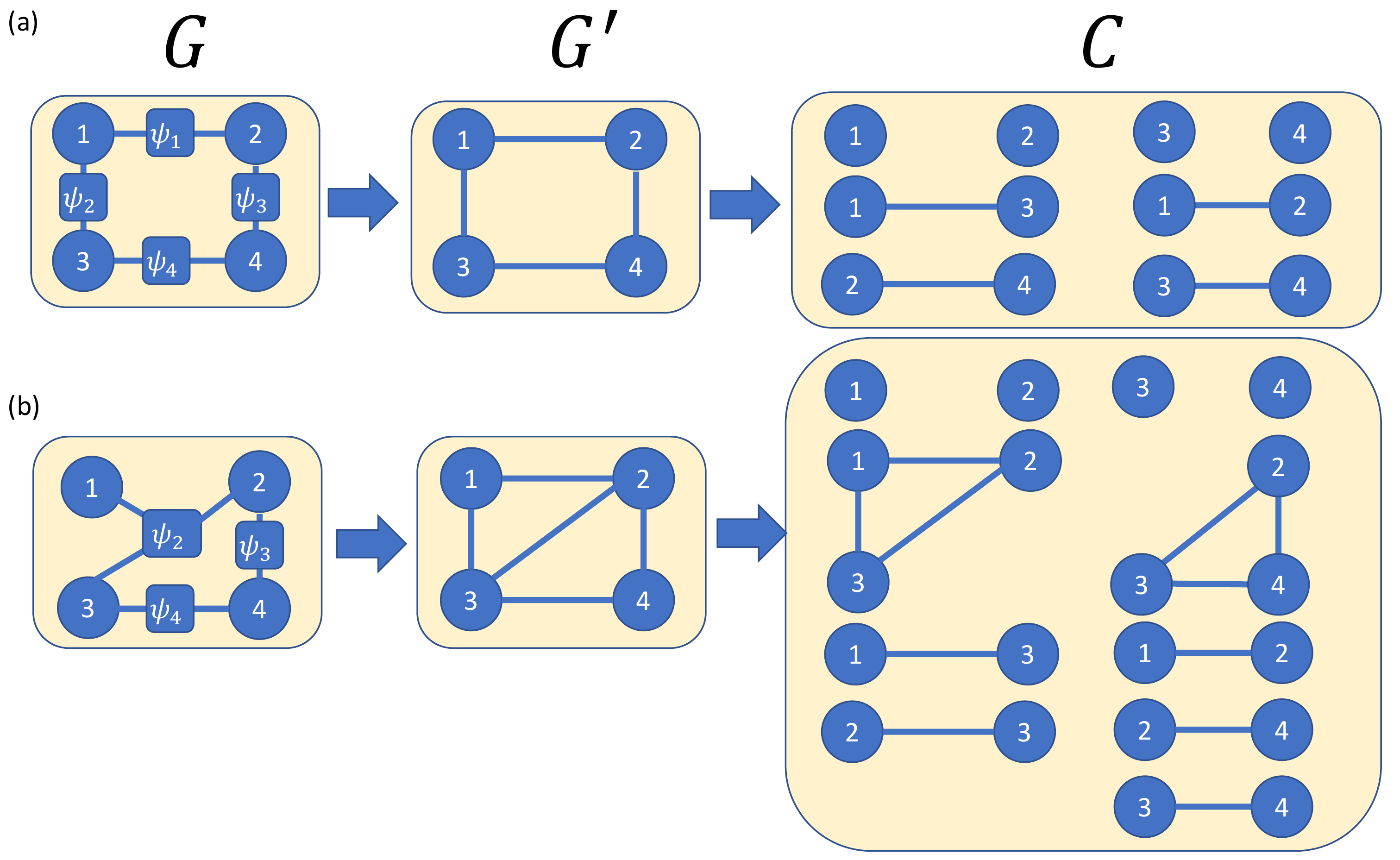}
    \caption{Here we illustrate the conversion of a factor graph into the set of relevant cliques for two different types of factor graph. In the first step each node touching a factor is connected to the nodes also touching that factor, and the factor removed. In the second we create a set of all of the induced subgraphs. 
    }
    \label{fig:factors}
\end{figure}
A particularly convenient form for the Hamiltonian $H$ is to express it as a sum of products of the variables nodes via the following construction. 
For every factor $k$ in the factor graph, we can replace the factor node and its neighborhood by the local complement in the graph. 
That is, all the variable nodes connected to the factor $k$ are connected into a complete graph (since they were previously nonadjacent), and the factor is disconnected from the graph. 
We call this modified factor graph (with the disconnected factor nodes now discarded) $G'$. 
Finally, let $C$ be the set of all cliques in $G'$, that is, the set of all complete induced subgraphs. 
We illustrate the transformation from $G$ to $G'$ and the set of cliques $C$ in \cref{fig:factors}.
With this notation, we have for some real numbers $J_b$ the general form
\begin{align}
\label{eq:cliqueHamiltonian}
    H(\boldsymbol{x}) = -\sum_{b \in C} J_b \prod_{a\in b} x_a\,.
\end{align}
The minus sign in front is purely a convention. 
Each of these terms functions as a local ``coupling'' that is local to the clique $b$ and lowers the energy by $J_b$ whenever all variables in $b$ are equal to 1. 

We can also see from this mapping that $H(\boldsymbol{0}) = -J_{\emptyset}$ is a constant ``energy shift''. 
Having already defined a partition function, this factor is redundant and could be set to zero. 
With that convention, we have the result that $p(\boldsymbol{0}) = \tfrac{1}{Z}\prod_k \exp(H(\boldsymbol{0})) = \tfrac{1}{Z}$. 
Alternatively, we could omit the partition function (setting $Z=1$) and keep $J_{\emptyset} = -\log p(\boldsymbol{0})$ as the ``free energy'' of the Gibbs ensemble. 
This second convention is sometimes useful for bookkeeping in calculations. 
Either way, this normalizing factor is likely hard to compute in general given the $J_b$. 
However, for the very small error rates in quantum computers, it might be possible to get a direct estimate of $p(\boldsymbol{0})$ from sampled data. 
Unfortunately, as the system size becomes large enough this eventually becomes exponentially small, and therefore impractical to estimate through sampling. 

Because of this, we relax our goal to specifying the entire probability distribution $p(\boldsymbol{x})$ \textit{up to normalization}, and this is now equivalent to specifying the values of the coupling constants $J_b$, one for each non-empty clique $b$ in $G'$. 
If there are at most $f$ factors with neighborhoods in $G$ of size at most $d$, then the entire probability distribution is specified by at most $f 2^d$ real numbers, which is much less than $2^n$ in the regime of interest, where $f$ is linear in $n$ and $d$ is constant. 

How might one instantiate these factors, i.e.\ learn the parameters $J_b$ given knowledge (or an assumption) about the factor graph $G$? 
First consider a set of variables $\boldsymbol{x}_r$ where $r \in C$. 
Suppose we consider an event where every variable outside of $r$, denoted $\boldsymbol{x}_{r^c}$, is zero. 
Then any product of $x_a$ that involves a variable outside of $r$ vanishes. 
That is, if $b \not\subseteq r$, then $\prod_{a\in b} x_a$ contains at least one 0.
The log-probability then greatly simplifies to
\begin{align}
    -\log p(\boldsymbol{x}_r,\boldsymbol{0}_{r^c}) &= \sum_{b\in C} J_b \prod_{a\in b} x_a \biggr|_{\boldsymbol{x}_{r^c}=\boldsymbol{0}}\\
    &= \sum_{b\subseteq r} J_b \prod_{a\in b} x_a\,.\label{eq:factorsum}
\end{align}
This gives a set of linear equations that relate the log-probabilities with the coupling constants when we enumerate over the $2^{|r|}$ bit-strings $\boldsymbol{x}_r$. 
To illustrate, suppose that $|r| = 2$, and write $p(x_1x_2,\boldsymbol{0})$ for $p(\boldsymbol{x}_{r},\boldsymbol{0}_{r^c})$. 
Then we would have the four equations
\begin{align}
\begin{array}{llllll}
-\log p(00,\boldsymbol{0}) & = & J_\emptyset\\
-\log p(01,\boldsymbol{0}) & = & J_\emptyset & + J_{01}\\
-\log p(10,\boldsymbol{0}) & = & J_\emptyset &          & + J_{10}\\
-\log p(11,\boldsymbol{0}) & = & J_\emptyset & + J_{01} & + J_{10} & + J_{11} \,.
\end{array}
\end{align}
Here we have adopted the convention that the label for the subset $b$ is just the bit string with 1 as the indicator function for set membership. 

We now make two important simplifications to this problem. 
The probabilities $p(\boldsymbol{x}_r,\boldsymbol{0}_{r^c})$ will be extremely small in general, even for independent noise. 
These cannot be efficiently learned from sampling, even if $|r|$ is a small constant, since they decay exponentially in $n$, the total number of bits. 
We therefore notice that we can use our relaxed goal of learning modulo the normalization to add $\log p(\boldsymbol{0}_{r_c})$ to both sides. 
By Bayes' theorem we have $p(\boldsymbol{x}_r,\boldsymbol{0}_{r_c}) = p(\boldsymbol{x}_r|\boldsymbol{0}_{r_c}) p(\boldsymbol{0}_{r_c})$, so the left hand side becomes a conditional probability, $-\log p(\boldsymbol{x}_r|\boldsymbol{0}_{r_c})$. 
Let us write $\partial r$ for the set of nodes in the neighborhood of~$r$. 
Then $\partial r$ forms a separating set, and using the global Markov property implies that the left hand side simplifies further to $-\log p(\boldsymbol{x}_r|\boldsymbol{0}_{\partial r})$. 
This boundary set $\partial r$ will have bounded size if each of the factors have bounded size, and if each variable participates in a bounded number of factors. 
Thus, the left hand side becomes conditional probability distributions over a bounded number of variables, for which an empirical estimation can be done with reasonable efficiency. 
To balance the right hand side, the term $J_\emptyset = -\log p(\boldsymbol{0})$ transforms to $J_{\boldsymbol{0}_r} := -\log p(\boldsymbol{0}_r|\boldsymbol{0}_{\partial r})$. 
Note that $J_{\boldsymbol{0}_r}$ does not appear anywhere in our Hamiltonian \cref{eq:cliqueHamiltonian}, so our estimate of this is in some sense a nuisance parameter. 

To summarize, we now have an equivalent system of equations where the quantities on the left hand side allow reasonable empirical estimates, given in illustration for $|r|=2$ by
\begin{align}
\begin{array}{llllll}
-\log p(00|\boldsymbol{0}_{\partial r}) & = & J_{00}\\
-\log p(01|\boldsymbol{0}_{\partial r}) & = & J_{00} & + J_{01}\\
-\log p(10|\boldsymbol{0}_{\partial r}) & = & J_{00} &          & + J_{10}\\
-\log p(11|\boldsymbol{0}_{\partial r}) & = & J_{00} & + J_{01} & + J_{10} & + J_{11} \,.
\end{array}
\end{align}
We have only to solve these equations for the $J_{\boldsymbol{x}_r}$ for $\boldsymbol{x}_r \not= \boldsymbol{0}_r$. 

With this binary ordering, we define a matrix 
\begin{align}
    A = \begin{pmatrix}
            1&0\\
            1&1
        \end{pmatrix}^{\otimes |r|}\,,
\end{align}
so that in general our equations have the simple form
\begin{align}
    -\log p(\boldsymbol{x}_r|\boldsymbol{0}_{\partial r}) = A J_{\boldsymbol{x}_r}\,.
\end{align}
Note that $A$ is invertible, with inverse
\begin{align}
    A^{-1} = \begin{pmatrix}
            1&0\\
            -1&1
        \end{pmatrix}^{\otimes |r|}\,.
\end{align}
Therefore, a general solution for the couplings in the subset $r$ is given by
\begin{align}
    J_{\boldsymbol{x}_r} = -A^{-1} \bigl(\log p(\boldsymbol{x}_r|\boldsymbol{0}_{\partial r})\bigr) \,.
\end{align}
By replacing the conditional probabilities by empirical estimates of conditional probabilities, we directly obtain empirical estimates for the coupling constants in any given clique $r$. 

A word on the theoretical precision of this estimation procedure. 
First, the condition number of $A$ is $\bigl(\tfrac{3+\sqrt{5}}{2}\bigr)^{|r|} \approx 2.61^{|r|}$, so the linear inversion becomes poorly conditioned for large $|r|$. 
Second, a useful estimate of $-\log p$ requires an estimate of $p$ to \textit{relative} precision $p (1\pm \epsilon)$ in order for the error to remain tolerable after the logarithm. 
This is expensive when $p$ is small. 
Despite these challenging theoretical limitations, we empirically observe quite favorable performance of this procedure. 
We leave open to future work the search for even better estimators with improved guarantees. 

The estimation problem substantially simplifies in the case where the factor graph is one-dimensional or, as in the case of \cref{fig:markovBlanket2}(b), is 1D up to some coarse graining. 
In that case, one only needs to estimate the marginals (as opposed to conditional marginals) on nearest-neighbors along the 1D geometry. 
This follows from a repeated application of Bayes' theorem together with the global Markov property, as we now show. 
Consider a chain of variables $\boldsymbol{x}_{1:n} = x_1x_2\ldots x_n$, which need not be binary as they may have arisen from some coarse graining. 
By applying Bayes' theorem followed by the global Markov property, we have
\begin{align*}
    p(\boldsymbol{x}_{1:n}) &= p(x_1|\boldsymbol{x}_{2:n}) p(\boldsymbol{x}_{2:n}) \\
    & = p(x_1|x_2) p(\boldsymbol{x}_{2:n}) \\
    & = \frac{p(x_1,x_2)}{p(x_2)} p(\boldsymbol{x}_{2:n}) \,.
\end{align*}
We can now recursively apply this idea to the rest of the chain starting at $x_2$ and we find
\begin{align*}
    p(\boldsymbol{x}_{1:n}) = \prod_{k=1}^{n-1}\frac{p(x_k,x_{k+1})}{p(x_{k+1})} \times p(x_{n}) \,.
\end{align*}
In this way, the case of a 1D structure allows some additional efficiency in estimating the model parameters since they can now be reconstructed from (estimates of) only nearest neighbor marginals.

In the case of the type of model proposed for the surface code, calculation of any particular factor will involve a maximum of 8 qubits. 
As previously discussed, for locally averaged noise, the joint distributions for all possible 8 qubit groupings can be ascertained with a single `RB-style' experiment. 
While this calculation is predicated on the fact that the underlying probability distribution does indeed obey the global Markov global property (which will not necessarily be the case for the actual noise in the device), ref~\cite{Abbeel2012} shows (and our numerics indicate that in an experimental setting) the process degrades gracefully, i.e.\ the calculated approximation will indeed approximate the underlying probability distribution. 
Note that although we can populate the factors, calculating the partition function (needed to normalise the graph so that it forms a probability distribution), still requires a calculation that scales exponentially. 
This is, however, not a problem when we don't need to know the probabilities but just the relative probabilities. 
Importantly, this means we can sample from the model in a scalable fashion in most cases of interest.

\section{Covariance bounds}\label{sec:covbound}

Suppose we have two finite probability distributions $p$ and $q$.
Now consider a random variable $X$, which is a real vector that is distributed according to either $p$ or $q$.
Let the covariance matrix of $X$ and the mean of $X$ with respect to the distribution $p$ be 
\begin{align}
	\Sigma_p &= \mathbb{E}_p\bigl[(X-\mu_p) (X-\mu_p)^T\bigr] \nonumber\\&= \mathbb{E}_p\bigl[X X^T\bigr] - \mu_p\mu_p^T \,,\quad \\&\quad\mu_p = \mathbb{E}_p[X]\,.
\end{align}
We would like to prove an upper bound on the operator norm of the difference, $\| \Sigma_p - \Sigma_q \|$, in terms of a distance between $p$ and $q$ and some notion of the length of the $X_j$.
We have the following theorem.
\begin{theorem}
Let $C = \mathrm{conv}\{X_j |  j \in S\}$ be the convex hull of the vectors $X_j$ in $S = \mathrm{supp}(p-q)$, the support of $p-q$.
Let $D = \mathrm{diam}(C)$ be the diameter of $C$ in the Euclidean distance and let $T = T(p,q) = \tfrac{1}{2}\|p-q\|_1$ be the statistical (1-norm) distance between $p$ and $q$.
Then we have 
\begin{align}
	\| \Sigma_p - \Sigma_q \| \le T\bigl(D^2 -\tfrac{1}{4}\|\mu_p-\mu_q\|^2\bigr) \,.
\end{align}
Moreover, the scaling $TD^2$ cannot be sharpened by any constant factor.
\end{theorem}

\begin{proof}
We introduce some new variables,
\begin{align}
	m &= \frac{p+q}{2}\,,\quad \\\mu &= \mathbb{E}_m[X] = \frac{\mu_p+\mu_q}{2}\,, \\
	\delta &= \frac{\mu_p-\mu_q}{2}\,,\quad \text{and}\\
    Y_j &= X_j - \mu\,,
\end{align}
so that $\mu_p = \mu+\delta$, $\mu_q = \mu-\delta$, and $\mathbb{E}_m[Y] = 0$.
In terms of these variables, we have
\begin{align}
	&\Sigma_p - \Sigma_q \nonumber\\
	&= \mathbb{E}_p\bigl[(Y-\delta) (Y-\delta)^T\bigr] - \mathbb{E}_q\bigl[(Y+\delta) (Y+\delta)^T\bigr] \nonumber\\
	&= \sum_j (p_j-q_j) Y_j^{} Y_j^T - \sum_j (p_j + q_j) (\delta Y_j^T +Y_j^{}\delta^T) \nonumber\\&\qquad+ \sum_j (p_j-q_j) \delta\delta^T\nonumber\\
	& = \sum_j (p_j-q_j) Y_j^{} Y_j^T \,.
\end{align}
In the last step, the middle terms vanish because $\mathbb{E}_m[Y] = 0$, and the last term vanishes identically because $\sum_j (p_j-q_j) = 0$.
Now we split the sum into two sets given by
\begin{align}
	j_+ = \{j:p_j>q_j\} \,,\quad j_- = \{j:p_j<q_j\} \,,
\end{align}
so that taking norms of both sides we have
\begin{align}
	&\|\Sigma_p - \Sigma_q \| \nonumber \\
        & = \biggl\| \sum_j (p_j-q_j) Y_j^{} Y_j^T \biggr\| \nonumber\\
	&= \biggl\| \sum_{j\in j_+} |p_j-q_j| Y_j^{} Y_j^T - \sum_{j\in j_-} |p_j-q_j| Y_j^{} Y_j^T \biggr\| \nonumber\\
	&= \|M_+ - M_-\|\,.
\end{align}
In the last equality, the matrices $M_\pm$ that we introduce are both positive semidefinite, $M_\pm \ge 0$, since they are positive sums of positive semidefinite matrices.
But for any two positive semidefinite matrices $M_\pm$, we have
\begin{align}
\label{eq:varopnormineq}
	&\| M_+ - M_-\| \nonumber\\&= \max_{v:\|v\|=1} |v^T(M_+ - M_-)v| \le \max\{\|M_+\| , \|M_-\|\}\,.
\end{align}
We will assume without loss of generality that $\|M_+\| \ge \|M_-\|$.
Because both $p$ and $q$ are probability distributions, we have
\begin{align}
\label{eq:1normplusminus}
	\sum_{j_+} |p_j-q_j| = \sum_{j_-} |p_j-q_j| = \frac{1}{2}\|p-q\|_1 = T(p,q)\,.
\end{align}
Using the inequality from Eq.~\ref{eq:varopnormineq}, the triangle inequality, H\"{o}lder's inequality, and Eq.~\ref{eq:1normplusminus} we have
\begin{align}
	\| \Sigma_p - \Sigma_q \| & \le \max\{\|M_+\| , \|M_-\|\} \\
	& \le \sum_{j\in j_+} |p_j-q_j| \, \bigl\|Y_j^{} Y_j^T \bigr\| \\
	& \le \sum_{j\in j_+} |p_j-q_j| \, \max_{k\in S} \bigl\|Y_k^{} Y_k^T \bigr\| \\
	& = T \, \max_{k\in S} \bigl\|X_k -\mu \bigr\|^2 \,.
\end{align}
Geometrically, $\max_{k\in S} \bigl\|X_k -\mu \bigr\|$ is the maximum distance between the vectors $X_k$ and some point $\mu$ which must be inside $C$, their convex hull.
This distance is clearly bounded by $D$, the diameter of $C$, and leads to the weaker bound of $\| \Sigma_p - \Sigma_q \|\le D^2 T$.

To get the sharper statement of the theorem, we need one more observation. 
Define $X^\star = X_{k^{\star}}$, where $k^\star = \arg\max_{k\in S} \bigl\|X_k -\mu \bigr\|$ is the index of any maximizer (if there is more than one, we choose one arbitrarily).
Because the diameter of $C$ is $D$, we must also have $r_p := \|X^\star -\mu_p \bigr\|\le D$ and $r_q := \|X^\star -\mu_q \bigr\|\le D$.
But the distance $\|\mu_p - \mu_q\|$ is fixed, so the distance of interest, $r := \|X^\star - \mu\|$, cannot be as large as $D$ unless $p=q$ as the distance to the midpoint will generally be shorter.
The maximizing distance will in fact be the median of the triangle formed by the three points $\mu_p, \mu_q, X^\star$.
Elementary geometry shows that this squared distance is then bounded as follows
\begin{align}
	&r^2 = \|X^\star - \mu\|^2 = \nonumber\\
	&\frac{2r_p^2+2r_q^2-\|\mu_p-\mu_q\|^2}{4} \le D^2-\frac{\|\mu_p-\mu_q\|^2}{4}\,.
\end{align}
This completes the proof of the stronger inequality.

To show that the upper bound cannot be strengthened to $c\, T D^2$ for any constant $c < 1$, consider the following example. 
For the probability distributions $p=(1,0)$, $q=(1-a,a)$ (with $1\ge a\ge0$), we have
\begin{align}
	T = \tfrac{1}{2}\|p-q\|_1 = a\,.
\end{align}
If we consider the scalar random variable $X_0=0$, $X_1=D$, then the covariance with respect to each probability distribution is just the variance.
Only $q$ yields a nontrivial variance given by $D^2 a (1-a)$, so the bound is indeed tight to leading order in $a$.
\end{proof}

In the case of reconstructing the ($n\times n$) Pauli covariance matrix, we have guaranteed convergence in the 1-norm between $p$ and an estimate $q$ of $p$, and the vectors $X_j$ have diameter $D = \sqrt{n}$, where $n$ is the number of qubits.
So the covariance estimate that we report using the initial density estimate is consistent, though it may have some bias.
Moreover, the scaling of the upper bound proven in the theorem is most likely an artifact in that case, for the following reason.
Both $p_j$ and $q_j$ get smaller when $j$ labels higher-weight errors, so $\|X_j\|$ gets suppressed for those values in the sum.
The bound is worst-case and is not sensitive to this.
It only takes a $1/w$ dependence with the weight to remove the scaling with $n$, so it is very likely that the $n$ dependence is not there in practice.

Finally we note it is not possible to upper bound the total variation distance between two finite probability distributions on bit strings by a function only of the covariance matrix of the random variable.

Consider $q$ being the uniform distribution on all strings of length $n = 2^k$ and $p$ being uniform on the $2 n$ strings that come from the rows of a Hadamard matrix and its complement. 
The TVD between $p$ and $q$ is easily computed to be $1 - n 2^{1-n}$. 
But the covariance between $p$ and $q$ is identically 0. 
In fact, the first and second moments are both zero.

\begin{table}[ht!]
\begin{tabular}{p{1cm}p{2cm}p{1.5cm}p{1.5cm}p{2cm}}
\toprule
Data Qubit&Alternative Qubits& CE \mbox{Alternative} & CE Ising & Ising Qubits \\ \midrule
6&[5, 7, 10, 11]&0.8166&0.8167&[5, 7, 2, 10]\\
13&[9, 14, 18]&0.874&0.875&[9, 14, 17]\\
18&[13, 14, 17]&0.797&0.799&[14, 17, 19]\\
\bottomrule
\end{tabular}
\caption{For each of the data qubits in the device (labeled in the order show in \cref{fig:markovBlanket}), the qubits that are linked to the qubit in question under the Ising model were read from the graph. 
All possible combinations of the same number of data qubits were assessed using conditional entropy (CE) and the qubits with the lowest value were found. 
Other than as listed above they were the same as the Ising model ansatz. 
We note that even in those cases, all but one of the qubits in each `optimal' blanket were the same as in the Ising model and the difference in the CE between the optimal qubits and the Ising ansatz were, in all cases, less than 0.5\%. 
Using a bootstrap methodology to ascertain error bars, shows us the 95\% confidence level in all cases is less than $\pm2\times 10^{-5}$.}\label{fig:cmi}
\end{table}

\section{Justifying the Ising Model with data}\label{sec:ce}

As noted in the text for this specific device the Ising model is unable to capture some of the longer range correlated errors that appear to exist in the device. 
This begs the question: is there another type of two-factor graph that we can draw (say, by way of example) connecting qubits 1 and 9 rather than 1 and 2) that might better represent the underlying probability distribution?

Given that we have access to the full distribution this is, in fact, a question we can attempt to  answer.

For (classical) random variables $\mathcal{X}, \mathcal{Y}$, the \textit{conditional entropy} (CE) of $\mathcal{Y}$ given $\mathcal{X}$ is defined as:
\begin{equation}
    H(\mathcal{Y}|\mathcal{X})=-\sum\limits_{x\in\mathcal{X},y\in\mathcal{Y}}p(x,y)\log\frac{p(x,y)}{p(x)}
\end{equation}
and quantifies the amount of information needed to describe the outcome of random variable $\mathcal{Y}$ given that we have knowledge of the value of $\mathcal{X}$. We can use this measure to determine the optimal qubits to link together in our graph. The lower the conditional entropy of $\mathcal{Y}$, the `stronger' the link to $\mathcal{X}$. If we specify a maximum number of links ($n$) we will entertain to any particular qubit ($q$). The for each candidate qubit ($q$) we can search to see which group of $n$ qubits will give us the lowest conditional entropy. 
Those are the qubits we might consider linking in our graph. 

If we decide to limit the number of `linked' qubits to be the same as in the Ising model (e.g.\ for a corner qubit we want to draw a link to two other qubits, for a middle qubit we will allow it to link to four other qubits) then a brute-force search for the best qubits to link together is quite tractable in the current device. 

Carrying out this procedure we find that the extracted data confirms that the Ising model links are the strongest links with three exceptions which are set out in \cref{fig:cmi}. 
Even in those exceptional cases, the `optimal' links are extremely close to the Ising model links, differing in all cases by only one qubit and a conditional entropy of less than 0.5\% from the Ising ansatz conditional entropy. 
This process was repeated 100 times with different distributions, bootstrapped from the original data (which gives us the the confidence level described in the table). 
The results were consistent, providing a large degree of confidence that the Ising model ansatz based on qubit physical location is a sensible ansatz to use.

\section{Reconstructed noise channels}\label{sec:theoretical}

As discussed in the main text, the Walsh-Hadamard transform can be used to move between the observed error probability domain and the Pauli channel eigenvalues (or, as here, the locally-averaged channel eigenvalues). 
Using this transform it is trivial to calculate the effect of multiple applications of the channel, even when one starts from an observed probability distribution. 
For instance two applications of the channel is calculated by squaring the eigenvalues. 
It is exactly this property that is used in fitting decay curves in randomised benchmarking.
 
Similarly one can easily model the effect of `partial' applications of the channel, e.g.\ by taking the square root of the eigenvalues one can calculate a channel that if applied twice would result in the original observed channel. 
This is most easily seen if one recalls that a super-operator representation of a Pauli channel (in a Pauli basis) is a diagonal matrix, with the diagonal elements being the Pauli eigenvalues of the channel. 
Taking the square root of these diagonal elements results in a channel that when applied twice (multiplied by itself) results in the original channel.

We can make this more precise if one assumes the noise is generated by a continuous time process. 
In that case, let $p$ be the probability distribution measured by our protocol and $W$ the Walsh-Hadamard transform, then combining \cref{eq:lt} and \cref{eq:pt} we have: 
\begin{equation}
\label{eq:continuouschannel}
p(t)=W^{-1} \exp(t \log(W p))
\end{equation}

\begin{figure*}[ht!]
\begin{tikzpicture}
    
\node[inner sep=0pt] at (-0,5) {\includegraphics[width=1\textwidth]{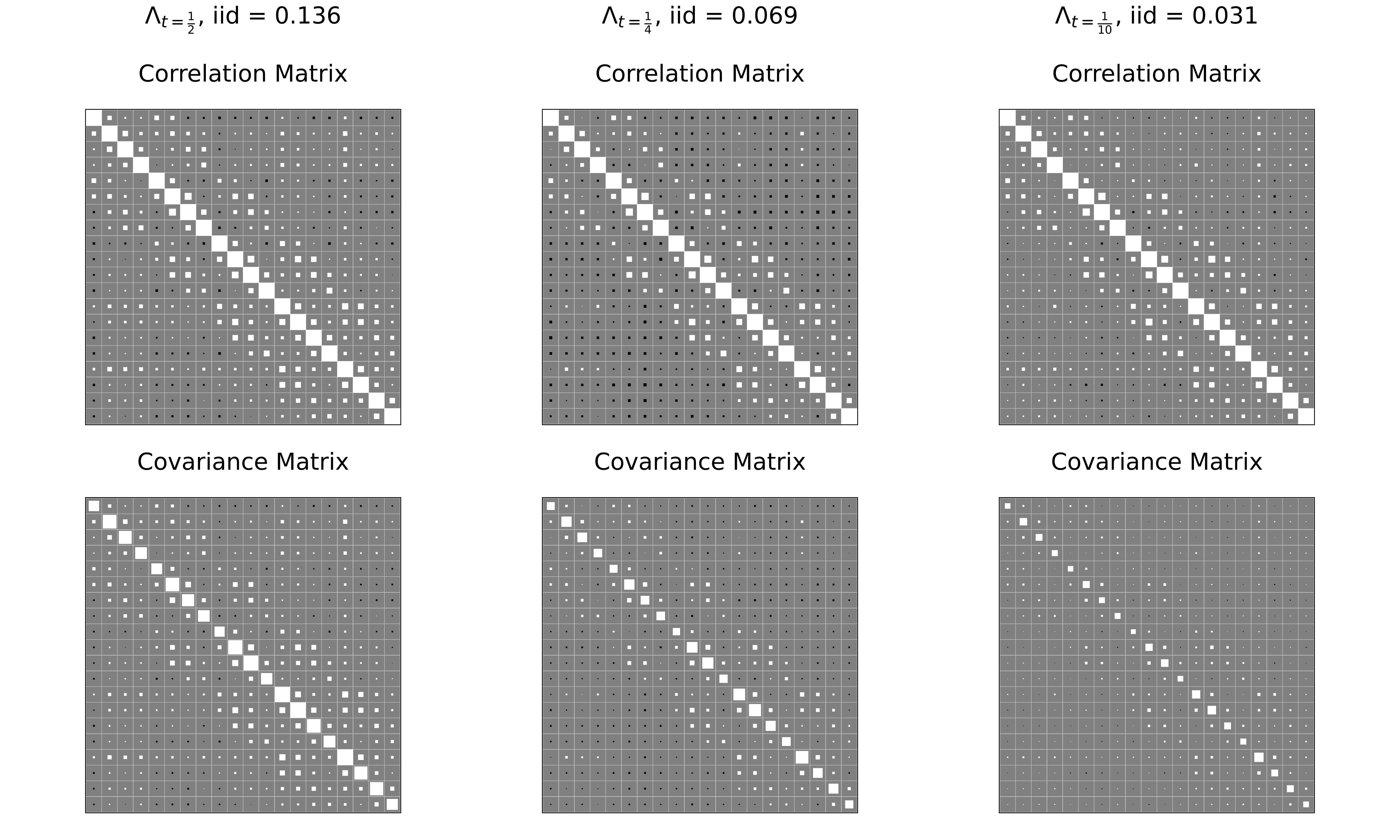}};
\node at (-8,10) {\textbf{(a)}};
\node[inner sep=0pt] at (-0,-4) {\includegraphics[width=0.8\textwidth]{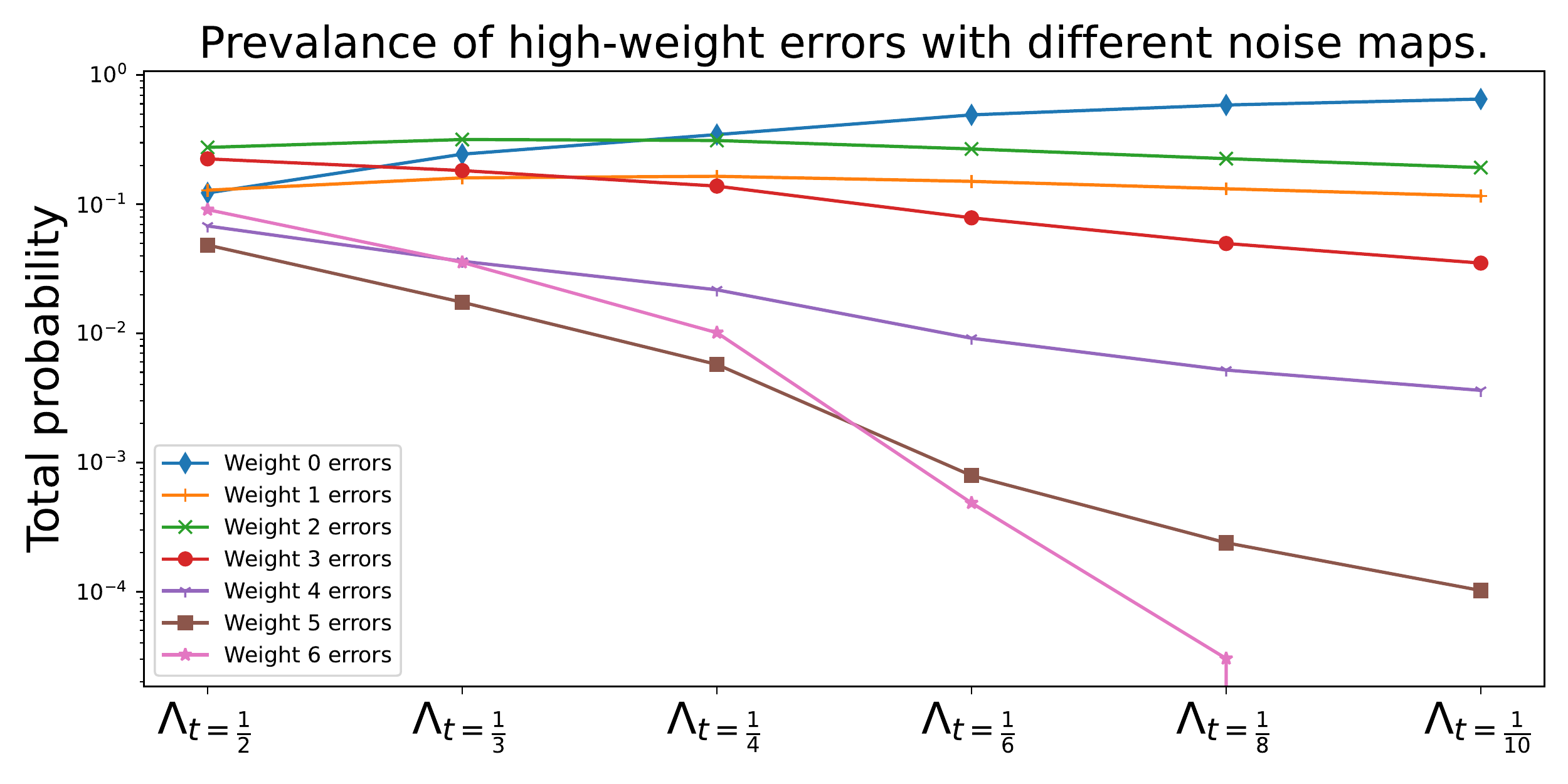}};
\node at (-8,-1) {\textbf{(b)}};

\end{tikzpicture}

\caption{(a) Hinton plots for the correlation and covariance matrices for some example reconstructed noise channels (see \cref{sec:theoretical}). 
The qubits are set out as Data Qubit 0 $\rightarrow$ Data Qubit 19. 
The value of the correlation/covariance is related to the area of the square. For the correlation matrices a full square has the value of 1. 
For instance, all squares on a diagonal are 1 (a full white square), as each qubit is correlated with itself. 
White represents positive correlations, black negative ones. For the covariance matrices, the scale is kept constant at a value of a full square = 0.18 (which is the largest value in all three example plots).
The first example $\Lambda_{t=\frac{1}{2}}$ shows the correlation matrix for the noise channel representing a single round of the ancilla preparation circuits. 
The title for each diagram also shows the average of the marginalised data qubit error rates (the number used in the simple iid model). 
The other diagrams show the same data for different values of $t$, e.g.\ the $\Lambda_{t=\frac{1}{4}}$ diagram represents the correlation matrix for the reconstructed channel that if applied twice would result in the same channel as $\Lambda_{t=\frac{1}{2}}$. 
As can be seen while the average qubit error rate decreases as $t\rightarrow 0$, the correlation pattern of each channel remains reasonably consistent and importantly retains the interesting noise features of the device. As might be expected the raw value of the covariance between the qubits decreases as the fidelity of the channel increases, but the patterns remain consistent.
(b) the total probability for errors with specific 'weights', i.e.\ affecting the specified number of qubits. As can be seen as $t$ decreases, there are fewer high weight errors. 
This is as anticipated, since as these higher fidelity maps are applied multiple times, lower weight errors can combine to form higher weight errors. 
The noise on the device is such, though, that errors with weights $\geq 3$ are still prevalent even on the least noisy maps.\label{fig:reconstructedCorrelations}}
\end{figure*}

As noted in the main text not every noise channel, including Pauli noise channels, is divisible in this fashion. 
The issue is that the object constructed in the manner of \cref{eq:continuouschannel} might not correspond to a completely positive trace preserving map. 
Indeed in practice, especially given the size of the channels and numeric imprecision, the raw transformed channel tended to have a number of small negative `probabilities'. 
As a practical matter we then took the step of projecting the resulting distributions onto the nearest probability simplex. 
It should be restated that the purpose of this exercise is not to generate channels that represented something \emph{achievable} in the device, but rather to `construct' counterfactual theoretical channels that have lower error rates than the observed distribution but that also, so far as possible, retain all the `interesting' features and correlations of the original measured noise. 
For instance \cref{fig:reconstructedCorrelations}(a) shows that these theoretical channels contain two-body correlations that appear broadly similar to the original observed channel and to each other. 

Finally we would anticipate that the higher fidelity noise maps (which represent partial applications of the observed noise map) would have fewer high-weight errors as a percentage of the total probability `budget'. 
This is because as they are applied multiple times, smaller weight errors can combine into larger weight errors, e.g.\ an $IX$ error and $XI$ error will combine to an $XX$ error. 
\Cref{fig:reconstructedCorrelations}(b) shows the prevalence of errors by weight in the various noise maps, confirming this behaviour.

Where our observed channel might not be below threshold (i.e.\ the logical error rate is above physical error rate), this will allow one to explore where (or if) the threshold is crossed, given the \emph{characteristics} of the noise.

\end{document}